%% file: main.tex
\documentclass{article}
\usepackage{ijcai20}

\usepackage{times}
\usepackage{soul}
\usepackage{url}
\usepackage[hidelinks]{hyperref}
\usepackage[utf8]{inputenc}
\usepackage[small]{caption}
\usepackage{graphicx}
\usepackage{amsmath}
\usepackage{amsthm}
\usepackage{booktabs}
\usepackage{algorithm}
\usepackage{algorithmic}
\usepackage{amsfonts}
\usepackage{comment}
\usepackage{xspace}
\usepackage{balance}

\newtheorem{theorem}{Theorem}
\newtheorem{lemma}{Lemma}
\newtheorem{claim}{Claim}
\theoremstyle{definition}		
\newtheorem{definition}{Definition}
\newtheorem{corollary}{Corollary}

\newcommand{\set}[1]{\left\{ #1 \right\}}
\newcommand{\rpar}[1]{\left( #1 \right)}

\newcommand{\abs}[1]{\left| #1 \right|}

\renewcommand{\natural}{\mathbb{N}}

\DeclareMathOperator{\rank}{rank}
\DeclareMathOperator{\poly}{poly}
\DeclareMathOperator{\election}{\mathcal{E}}
\DeclareMathOperator{\tie}{\succ_{\mathsf{tie}}}
\DeclareMathOperator{\tieprime}{{\succ}^\prime_{\mathsf{tie}}}
\DeclareMathOperator{\committees}{\mathbb{C}}
\DeclareMathOperator{\ps}{\pi}

\def\eqdef{\mathrel{{:}{=}}}
\def\e#1{\emph{#1}}

\def\T{\mathbf{T}}
\def\P{\mathbf{P}}
\def\Q{\mathbf{Q}}
\def\M{\mathbf{M}}

\def\neccom{\mathsf{NecCom}}
\def\poscom{\mathsf{PosCom}}
\def\necmem{\mathsf{NecMem}}
\def\posmem{\mathsf{PosMem}}
\def\angs#1{\mathord{\langle{#1\rangle}}}
\def\neccompk{\neccom\angs{k}}
\def\poscompk{\poscom\angs{k}}
\def\necmempk{\necmem\angs{k}}
\def\posmempk{\posmem\angs{k}}

\def\NW{\text{$\mathsf{NW}$}\xspace}
\def\PW{\text{$\mathsf{PW}$}\xspace}
\def\NTW{\text{$\mathsf{NTW}$}\xspace}
\def\PTW{\text{$\mathsf{PTW}$}\xspace}
\def\NTWk{\text{$\mathsf{NTW}\angs{k}$}\xspace}
\def\PTWk{\text{$\mathsf{PTW}\angs{k}$}\xspace}

\def\NWS{\text{$\mathsf{NWS}$}\xspace}
\def\PWS{\text{$\mathsf{PWS}$}\xspace}

\def\mypar#1{\par\vskip0.5em\noindent\underline{#1}\,}

\def\deg{\mathrm{deg}}

\newcommand{\eat}[1]{}

\newenvironment{hintproof}{\paragraph{Proof (sketch):}}{\hfill$\square$}

\newif\ifLong
\Longtrue

\title{The Complexity of Determining the Necessary and Possible\\ Top-k
  Winners in Partial Voting Profiles}

\author{
Aviram Imber
\and
Benny Kimelfeld
\affiliations
Technion -- Israel Institute of Technology
\emails
\{aviram.imber, bennyk\}@cs.technion.ac.il
}

\begin{document}

\maketitle

\begin{abstract}
 When voter preferences are known in an incomplete (partial) manner, winner determination is commonly treated as the identification of the necessary and possible winners; these are the candidates who win in all completions or at least one completion, respectively, of the partial voting profile. In the case of a positional scoring rule, the winners are the candidates who receive the maximal total score from the voters.  Yet, the outcome of an election might go beyond the absolute winners to the top-$k$ winners, as in the case of committee selection, primaries of political parties, and ranking in recruiting. We investigate the computational complexity of determining the necessary and possible top-$k$ winners over partial voting profiles. Our results apply to general classes of positional
  scoring rules and focus on the cases where $k$ is given as part of
  the input and where $k$ is fixed.
\end{abstract}

\section{Introduction}

\input{intro.tex}

\input{preliminaries.tex}

\input{input-k.tex}

\input{bounded.tex}

\input{conclusions.tex}

\balance

\bibliographystyle{named}
\bibliography{ijcai20}

\end{document}

\section{Necessary and Possible Committee Members}

We start with some notations. For a set $A$ and a partition
$A_1, ..., A_t$ of $A$, let $P(A_1, ..., A_t)$ denote the partial
order
$$ P(A_1, ..., A_t) \eqdef \set{a_1 \succ ... \succ a_t : \forall i \in [t], a_i \in A_i}$$
and let $O(A_1, ..., A_t)$ denote an arbitrary linear order on $A$ that completes $P(A_1, ..., A_t)$. A linear order $a_1 \succ ... \succ a_t$ is also denoted as a vector $(a_1, ..., a_t)$. A \e{concatenation} of two linear orders is $(a_1, ..., a_t) \circ (b_1, ..., b_\ell) = (a_1, ..., a_t, b_1, ..., b_\ell)$. Similarly, a concatenation of two voting profiles is $(P_1^1, ..., P_1^t) \circ (P_2^1, ..., P_2^\ell) = (P_1^1, ..., P_1^t, P_2^1, ..., P_2^\ell)$.

In some proof we use the technique of circular voting blocks from \cite{DBLP:conf/atal/BaumeisterRR11}. For a set $A = \set{a_1, ..., a_t}$, for every $i \in [t]$ let the $i$th circular vote be
\begin{align*}
    M_i(A) \eqdef (a_i, a_{i+1}, ..., a_t, a_1, a_2, ..., a_{i-1})
\end{align*}

\subsection{Varying Committee Size}

\begin{definition}[X3C]
Given $V = \set{v_1, ..., v_{3q}}$ and a collection $E$ of 3-element subsets of $V$, we are asked whether we can cover all the elements of $V$ using non-overlapping sets in $E$.
\end{definition}
For every vertex $v \in V$, we denote the edges of $v$ by $E(v) = \set{e \in E : v \in e}$.

\begin{theorem}
  \label{thm:PluralityNecMem}
  $\necmem$ is coNP-complete under the plurality rule.
\end{theorem}
\begin{proof}
By reduction from X3C. Given $V,E$ we construct an election instance $\election = (C, P, \tie)$ with respect to plurality where $C = E \cup \set{c^*}$, $\tie = O(E, \set{c^*})$ and the partial profile is $P = (P_1, ..., P_{3q}, T_{3q+1}, T_{3q+2}, T_{3q+3})$. For every $i \in [3q]$,
$$ P_i = P (E(v), E \setminus E(v), \set{c^*})$$
This means that the $i$th voter can vote only for edges that cover $v_i$ (assume that every vertex in $V$ can be covered by at least one edge, otherwise the problem is trivial). For $ i > 3q$ the order is $T_i = O(\set{c^*}, E)$ (these orders are complete). We show that the X3C has a solution if and only if $c^*$ is not a necessary $q$-committee member. \par
Suppose $c^*$ is not a necessary $q$-committee member, that is, there are $q$ candidates $e_{i_1}, ..., e_{i_q}$ and a completion $T$ of $P$ such that $s(T, e_{i_j}) \geq s(T,c^*) = 3$ for every $j \in [q]$. Since each edge can get at most 3 votes we get that $s(T, e_{i_j}) = 3$ for every $j \in [q]$. Therefore $e_{i_1}, ..., e_{i_q}$ is an exact cover: every vertex $v_i$ is covered by the edge that $P_i$ voted for, and the edges are not overlapping (if two edges are overlapping then one of them gets at most 2 votes). \par
Conversely, given an X3C solution $Q = \set{e_{i_1}, ..., e_{i_q}}$ define a profile $T = (T_1, ..., T_{3q+3})$ such that for every $i \in [3q]$,
$$ T_i = O(Q \cap E(v_i), E(v_i) \setminus Q, E \setminus E(v_i), \set{c^*})$$
Every $T_i$ extends $P_i$, $s(T, e) = 3$ for every $e \in Q$ and $s(T, c^*) = 3$. Hence all edges in $Q$ defeat $c^*$, and $c^*$ is a not a necessary $q$-committee member.
\end{proof}

\begin{definition}[Dominating Set]
  Given an undirected graph $G = (V,E)$ and an integer $k$, we are asked whether there exists a subset $D \subseteq V$ of size $k$ such that every vertex is either in $D$ or is adjacent to some vertex in $D$.
\end{definition}

\begin{theorem}
\label{thm:PluralityPosMem}
$\posmem$ is NP-complete under the plurality rule and is W[2] hard with the parameter $k$.
\end{theorem}
\begin{proof}
By reduction from dominating set. Given a graph $(V = \set{v_1, ..., v_n}, E)$ we construct an election $\election = (C, P, \tie)$ with respect to plurality where $C = V \cup \set{c^*}$, $\tie = O(\set{c^*}, V)$ and $P = (P_1, ..., P_n)$. Let $N(v_i)$ be the set of neighbours of $v_i$ in the graph, and let $N(v_i)^* =  N(v_i) \cup \set{v_i}$. For every $i \in [n]$,
$$ P_i = P(N(v_i)^*, V \setminus N(v_i)^*, \set{c^*}) $$
This means that the $i$th voter can vote only for vertices that dominate $v_i$. We show that the graph has a dominating set of size $k$ if and only of $c^*$ is a possible $k$-committee member. \par
Suppose there is a dominating set $D = \set{v_{i_1}, ..., v_{i_k}}$, consider the completion
$$ T_i = O(N(v_i)^* \cap D, N(v_i)^* \setminus D, V \setminus N(v_i)^*, \set{c^*})$$
In this completion for every $v \notin D$ we get $s(T, v) = 0$. These are $n-k$ candidates that $c^*$ defeats, therefore $c^*$ is a possible $k$-committee member. Conversely, if $c^*$ is a possible $k$-committee member then in some completion $T$ it defeats at least $n-k$ candidates, and these candidates have a score 0 in $T$. Let $D$ be the set of candidates that $c^*$ does not defeat in $T$, all voters voted for candidates in $D$ and $|D| \leq k$. A voter $P_i$ can only vote for vertices which dominate $v_i$, hence $D$ is a dominating set of size at most $k$. \par
The problem of dominating set is also W[2] hard with the parameter $k$, so this reduction implies that $\posmem$ is W[2]-hard with the parameter $k$ under the plurality rule.
\end{proof}

\begin{definition}
Given a partial order $P_i$, the \e{reversed order} is $P_i^R = \set{x \succ y : (y \succ x) \in P_i}$. 
\end{definition}
Note that $T_i$ extends $P_i$ if and only if $T_i^R$ extends $P_i^R$.

\begin{definition}
Given a binary positional scoring rule $\Vec{s}_m$, the \e{complementary-reversed} scoring rule is $\Vec{s}_m^R$ such that for every $i \in [m]$, $\Vec{s}_m^R(i) = 1-\Vec{s}_m(m+1-i)$.
\end{definition}
For example, the complementary-reversed rule of $t$-approval is $t$-veto, and specifically the complementary-reversed rule of plurality is veto.

\begin{lemma}
\label{lemma:reversing}
For every binary positional scoring rule $\Vec{s}_m$ there is a reduction
\begin{enumerate}
    \item from $\necmem$ under $\Vec{s}_m$ to $\overline{\posmem}$ under $\Vec{s}_m^R$.
    \item from $\posmem$ under $\Vec{s}_m$ to $\overline{\necmem}$ under $\Vec{s}_m^R$.
\end{enumerate}
\end{lemma}
\begin{proof}
Given $\election^1 = (C, \Vec{s}_m, P^1, \tie)$ where $P^1 = (P_1, ..., P_n)$ consider $\election^2 = (C, \Vec{s}_m^R, P^2, \tie^R)$ where $P^2 = (P_1^R, ..., P_n^R)$. The total scores of the candidates in $\election^b$ are denoted by $s^b$. Let $T^1 = (T_1, ..., T_n)$ be a completion of $P^1$, observe the completion $T^2 = (T_1^R, ..., T_n^R)$ of $P^2$. For every candidate $c$ and a voter $v_i$ we get $s^2(T_i^R, c) = 1-s^1(T_i, c)$, so overall $s^1(T^2, c) = n-s^2(T^1, c)$. Since the tie-breaking order is also reversed, it holds that $\rank(T^2, c) = m+1-\rank(T^1, c)$ for every $c \in C$. In same way, if $T^2$ is a completion of $P^2$ then reversing the orders gives a completion $T^1$ of $P^1$ such that $\rank(T^1, c) = m+1-\rank(T^2, c)$ for every $c \in C$. We can deduce that for any candidate $c$ and for any integer $k$,
\begin{enumerate}
    \item There exists a completion $T^1$ of $P^1$ such that $\rank(T^1, c) \geq k+1$ ($c$ is not a necessary $k$-committee member in $\election^1$) if and only if there exists a completion $T^2$ of $P^2$ such that $\rank(T^2, c) \leq m-k$ ($c$ is a possible $(m-k)$-committee member in $\election^2$).
    \item There exists a completion $T^1$ of $P^1$ such that $\rank(T^1, c) \leq k$ ($c$ is a possible $k$-committee memeber in $\election^1$) if and only if there exists a completion $T^2$ of $P^2$ such that $\rank(T^2, c) \geq m-k+1$ ($c$ is not a necessary $(m-k)$-committee member in $\election^2$).
\end{enumerate}
\end{proof}

We can use lemma \ref{lemma:reversing} and theorems \ref{thm:PluralityNecMem}, \ref{thm:PluralityPosMem} to deduce hardness results for the veto rule.
\begin{corollary}
\label{cor:NecPosMemVeto}
$\necmem$ is coNP-complete and $\posmem$ is NP-complete under the veto rule.
\end{corollary}

In addition, it is not hard to see that for any candidate $c$, choosing $k=1$ and $\tie = O(\set{c}, C \setminus \set{c})$ in the proof of lemma \ref{lemma:reversing} implies that $c$ is a PW in $\election^1$ if and only if $c$ is not a necessary $(m-1)$-committee member in $\election^2$. Since PW is NP-complete under every pure binary positional scoring rule other than plurality and veto, we get following result (combining with Theorem \ref{thm:PluralityNecMem} and Corollary \ref{cor:NecPosMemVeto}).

\begin{corollary}
\label{cor:NecMemBinary}
$\necmem$ is coNP-complete under every pure binary positional scoring rule.
\end{corollary}

\begin{definition}
A positional scoring rule $\Vec{u}_m$ \e{polynomially contains} another rule $\Vec{v}_m$ if there exist a polynomial $p(m)$, an index $i_m \leq p(m)-m$ and $a_m>0, b_m$ which satisfy $\Vec{v}_{m}(j) = a_m \cdot \Vec{u}_{p(m)}(i_m+j) + b_m$ for every $m, j \in [m]$.
\end{definition}
For example, plurality is polynomially contained in $t$-approval for any fixed $t$, by choosing $p(m) = m+t-1, i_m = t-1, a_m = 1, b_m = 0$. Similarly, veto is polynomially contained in $t$-veto.

\begin{lemma}
\label{lemma:polyContained}
If a positional scoring rule $\Vec{s}^2$ polynomially contains $\Vec{s}^1$ then there is a reduction
\begin{enumerate}
    \item from $\necmem$ under $\Vec{s}^1$ to $\necmem$ under $\Vec{s}^2$.
    \item from $\posmem$ under $\Vec{s}^1$ to $\posmem$ under $\Vec{s}^2$. 
\end{enumerate}
\end{lemma}
\begin{proof}
Let $p(m), i_m, a_m, b_m$ be the functions from the definition of polynomially contained vectors. Given $\election^1 = (C^1, P^1, \Vec{s}^1_m, \tie^1)$ where $P^1 = (P_1^1, ..., P_n^1)$ consider $\election^2 = (C^2, P^2, \Vec{s}_{p(m)}^2, \tie^2)$ where:
\begin{itemize}
    \item $C^2 = C^1 \cup W^1 \cup W^2$ where $W^1, W^2$ are disjoint sets of new candidates of sizes $ |W^1| = i_m, |W^2| = p(m)-i_m-m $. Note that $|C^2| = p(m)$.
    \item The voting profile is $P^2 = (P_1^2, ..., P_n^2)$. In every $P_j^2$ the $i_m$ highest ranked candidates are $W^1$, the $p(m)-i_m-m$ lowest ranked candidates $W^2$, and between these two sets the candidates $C^1$ are the same as in $P_j^1$. Formally, $P_j^2 = P_j^1 \cup P(W^1, C^1, W^2)$.
    \item In the tie-breaking order, the $i_m$ highest ranked candidates are $W^1$, the $p(m)-i_m-m$ lowest ranked candidates $W62$, and between these two sets the candidates $C^1$ are the same as in $\tie^1$. Formally, $\tie^2 = O(W^1) \circ \tie^1 \circ O(W^2)$.
\end{itemize}
Let $T^1 =(T_1^1, ..., T_n^1)$ be a completion of $P^1$. Observe the completion $T^2 =(T_1^2, ..., T_n^2)$ of $P^2$ where $T^2_j = O(W^1) \circ T^1_j \circ O(W^2)$. Since $s^1$ is polynomially contained in $s^2$, for every $c \in C^1$ we get that $s^2(T_j^2, c) = a_m \cdot s^1(T_j^1, c) + b_m$ so overall $s^2(T^2, c) = a_m \cdot s^1(T^1, c) + n b_m$ and $\rank(T^2, c) = \rank(T^1, c)+i_m$. Conversely, let $T^2$ be a completion of $P^2$, every $T^2_j$ has to be of the form $O(W^1) \circ O(C^1) \circ O(W^2)$ so removing $W^1, W^2$ from all orders gives a completion $T^1$ of $P^1$ such that $\rank(T^1, c) = \rank(T^2, c) - i_m$ for every $c \in C^1$. These arguments show that for every $c \in C^1$ and $k$, 
\begin{enumerate}
    \item $c$ is a necessary $k$-committee member in $\election^1$ if and only $c$ is a necessary $(k+i_m)$-committee member in $\election^2$.
    \item $c$ is a possible $k$-committee member in $\election^1$ if and only $c$ is a possible $(k+i_m)$-committee member in $\election^2$.
\end{enumerate}
\end{proof}

\begin{definition}
  A scoring rule $\Vec{s}_m$ is \e{polynomially frequent} if there exists a score $x$ and a constant $\varepsilon$ such that $|\set{i : \Vec{s}_m(i) = x}| = \Omega(m^\varepsilon)$.
\end{definition}

\begin{theorem}
$\necmem$ is coNP-complete and $\posmem$ is NP-complete with respect all polynomially frequent pure positional scoring rules.
\end{theorem}
\begin{proof}
Let $m_x$ be the minimal $m$ for which $x$ appears in $\Vec{s}_m$, and let $i_x$ be the index of $x$ in $\Vec{s}_{m_x}$. Consider two cases. First, if $i_x > 1$, then there exist some score $y_m > x$ such that $\Vec{s}_{O(m^{1/\varepsilon})}$ contains the vector $(y_m, x, ..., x)$ of length $m+1$. Choosing $p(m) = O(m^{1/\varepsilon})$, $i_m$ as the index of $y_m$, $a_m = y-x, b_m = x$ shows that in this case plurality is polynomially contained in $\Vec{s}_m$. The hardness results are then derived by theorems \ref{thm:PluralityNecMem}, \ref{thm:PluralityPosMem} and lemma \ref{lemma:polyContained}. \par
Second, if $i_x = 1$, then there exist some score $y_m < x$ such that $\Vec{s}_{O(m^{1/\varepsilon})}$ contains the vector $(x, ..., x, y_m)$ of length $m+1$. Choosing $p(m) = O(m^{1/\varepsilon})$, $i_m$ as the smallest index of $x$, $a_m = x-y, b_m = y$ shows that in this case veto is polynomially contained in $\Vec{s}_m$. The hardness results are then derived by corollary \ref{cor:NecPosMemVeto} and lemma \ref{lemma:polyContained}.
\end{proof}

\begin{theorem}
\label{thm:BordaNecMem}
$\necmem$ is coNP-complete under the Borda rule.
\end{theorem}

\begin{proof}
By reduction from X3C. Given an X3C instance $V = \set{v_1, ..., v_{3q}}, E = \set{e_1, ..., e_m}$ we construct an election $\election = (C, P, \tie)$ with respect to Borda. The candidates are $C = E \cup W \cup \set{c^*}$ where $W = \set{w_1, ..., w_{m-1}}$ and $\tie = O(E, \set{c^*}, W)$. The voting profile is composed of three parts $P = P^1 \circ T^2 \circ T^3$:
\begin{itemize}
    \item $P^1 = (P_1^1, ..., P_{3q}^1)$. For every $i \in [3q]$, only edges which cover $v_i$ can receive a score of $2m-1$, edges which do not cover $v_i$ can receive at most $m-1$, and $c^*$ receives 0. Formally, denote by $d(v_i) = |E(v_i)|$ the degree of $v_i$ in the graph and let $W_{\leq j} = \set{w_1, ..., w_j}, W_{> j} = \set{w_{j+1}, ..., w_{m-1}}$. We assume that $1 \leq d(v_i) \leq m-1$, otherwise the problem is trivial. The partial order is
    \begin{align*}
        P_i^1 = P( & E(v_i), W_{\leq m-d(v_i)}, \\
        & W_{> m-d(v_i)} \cup (E \setminus E(v_i)), \set{c^*})
    \end{align*}
    
    Note that for every $e \in E \setminus E(v_i)$, the number of candidates which are ranked above $e$ in $P_i$ is $|E(v_i)| + m-d(v_i) = m$, so indeed it receives a score of at most $m-1$.
    \item $T^2$ is composed of $S_{cover} \eqdef 3(2m-1) + 3 \sum_{i=1}^{q-1} (m-i)$ copies of the profile $(T^2_1, ..., T^2_m)$. For every $i \in [m]$ the order $T^2_i$ is constructed in the following way. Start with $M_i(E)$, then insert $W$ and $c^*$ such that the score of $c^*$ is $m$ and scores of $e_i, ..., e_m, e_1, ..., e_{i-1}$ are $2m-1, 2m-3, ..., 3, 0$ accordingly if $m$ is even. If $m$ is odd then the scores of the edges are the same as the before but with $m-1$ instead of $m$ and 1 instead of 0. Note that in both cases, the sum of the scores of the edges is the same.
    \item $T^3$ is composed of $4(3q + S_{cover})$ copies of the profile $(T^3_1, ..., T^3_{m+1})$. For this part only, denote $e_{m+1} = c^*$. For every $i \in [m+1]$, $T^3_i = M_i(\set{e_1, ..., e_{m+1}}) \circ O(W)$.
\end{itemize}
Now we state some observations regarding the profile. In $T^2$, the score of $c^*$ is
\begin{align*}
    s(T^2, c^*) = S_{cover} \cdot s((T^2_1, ..., T^2_m), c^*) =  S_{cover} \cdot m^2
\end{align*}
For every $e \in E$, the score is
\begin{align*}
    s(T^2, e) &= S_{cover} \sum_{i=2}^{m} (2i-1) =  S_{cover} \rpar{m^2 - 1}\\
    &= s(T^2, c^*) - S_{cover}
\end{align*}
In $T^3$, for every $w \in W$, the score is
\begin{align*}
    s(T^3, w) &\leq 4(3q + S_{cover}) \cdot m(m-2) \\
    &\leq 4m^2 (3q + S_{cover})
\end{align*}
For every $c \in C \setminus W$, the score is
\begin{align*}
    s(T^3, c) &= 4(3q + S_{cover}) \sum_{i=1}^{m+1} (2m-i) \\
    &= 2(3q + S_{cover})(3m^2+m-2) \\
    &\geq 6m^2 (3q + S_{cover})
\end{align*}
Hence for every pair $w \in W, c \in C \setminus W$ it holds that $s(T^3, c) - s(T^3, w) \geq 2m^2 (3q + S_{cover})$. Let $T = T^1 \circ T^2 \circ T^3$ be a completion of $P$. For every pair $w \in W, c \in C \setminus W$ we know that $s(T^1, w) - s(T^1, c) < 6qm$, $s(T^2, w) - s(T^2, c) < S_{cover} \cdot 2m^2$ and $s(T^3, c) - s(T^3, w) \geq 2m^2 (3q + S_{cover})$. Therefore $s(T, c) - s(T, w) > 0$, which means that the candidates in $C \setminus W$ always defeat all candidates in $W$.

\begin{claim}
If the graph has an X3C then $c^*$ is not a necessary $q$-committee member.
\end{claim}
\begin{proof}
Assume wlog that the exact cover is $Q = \set{e_1, ..., e_q}$ and every edge in the cover is $e_i = \set{i, q+i, 2q+i}$. Define a completion $T = T^1 \circ T^2 \circ T^3$, $T^1 = (T^1_1, ..., T^1_n)$. In $T^1$, every edge in the cover $e_i \in Q$ gets the following scores:
\begin{itemize}
    \item $e_i$ gets $2m-1$ from $T^1_i, T^1_{q+i}, T^1_{2q+i}$, that is, $e_i$ is placed at the top of the vertices which it covers.
    \item $e_i$ gets $m-1$ from the vertices of $e_{i+1}$, gets $m-2$ from the vertices of $e_{i+2}$ and so on (when we reach $e_q$ we go to $e_1$ and continue until $e_{i-1}$). Note that in this way there is no vertex that should give the same score to two different edges, and since $q \leq m$ the range of scores $e_i$ gets is $\set{m-1, m-2,  ..., m-q+1} \subseteq \set{m-1, ..., 1}$ as required by the definition of $P^1$.
\end{itemize}
For every edge $e \in Q$ the total score in $T^1$ is $s(T^1, e) = 3(2m-1) + 3 \sum_{i=1}^{q-1} (m-i) = S_{cover}$ and recall that $s(T^1, c^*) = 0$. Combining this with what we already showed for $T^2, T^3$ implies that
\begin{align*}
    s(T, e) - s(T, c^*) & = S_{cover} - S_{cover} + 0 = 0
\end{align*}
Overall, all candidates in $Q$ defeat $c^*$, hence $c^*$ is not a necessary $q$-committee member.
\end{proof}

\begin{claim}
If $c^*$ is not a necessary $q$-committee member then the graph has an X3C.
\end{claim}
\begin{proof}
Let $T = T^1 \circ T^2 \circ T^2$ be a completion where at least $q$ candidates defeat $c^*$, let $Q$ be the $q$ highest rated  candidates in $T$. $c^*$ always defeats all candidates in $W$ hence $Q \subseteq E$. We show a lower bound and an upper bound on the total score of $Q$ in $T^1$.
\paragraph{Lower bound} As we already showed, every $e \in Q$ should get $s(T^1, e) \geq S_{cover}$ in order to defeat $c^*$, therefore
\begin{align*}
    \sum_{e \in Q} s(T^1, e) &\geq q \cdot S_{cover} = 3q \cdot \rpar{2m-1 + \sum_{i=1}^{q-1} (m-i)} \\
    & = 3mq^2 + 3mq - \frac{3q^3}{2} + \frac{3q^2}{2} - 3 q
\end{align*}
\paragraph{Upper bound} For every $v \in V$ denote by $d_Q(v)$ the degree of $v$ in the sub-graph induced by $Q$, it holds that $\sum_{v \in V} d_Q(v) = 3q$. We get
\begin{align*}
    \sum_{e \in Q} s(T^1, e) &\leq \sum_{i=1}^{3q} \rpar{\sum_{j=1}^{d_Q(v_i)} (2m-j) + \sum_{j=1}^{q-d_Q(v_i)} (m-j)} \\
    & = 3mq^2 + 3mq - \frac{3q^3}{2} + \frac{3q^2}{2} - \sum_{i=1}^{3q} d_Q(v_i)^2 
\end{align*}
Overall, the two bounds imply that $\sum_{i=1}^{3q} d_Q(v_i)^2 \leq 3q$ and this is possible only if all the degrees are 1: When all degrees are 1 the sum is exactly $3q$. If we decrease $d_Q(v_i)$ to zero and increase $d_Q(v_j)$ to 2 then the total sum increases by 3. Any further changes also increase the total sum. Therefore all degrees in the sub-graphs induced by $Q$ are 1, and $Q$ is an X3C.

\end{proof}

\end{proof}

\subsection{Fixed Committee Size}

\begin{definition}
Given a positional scoring rule, a profile $P = (P_1, ..., P_n)$ and alternatives $W = \set{w_1, ..., w_k}$ define the set of possible scores:
$$ \ps_W (P) = \set{(s(T,w_1), ..., s(T,w_k)) \mid T \text{ completes } P} $$
\end{definition}
Note that $\ps_W (P) \subseteq [n \cdot \Vec{s}_m(1)]_0^k$ where $[n \cdot \Vec{s}_m(1)]_0 = \set{0, ..., n \cdot \Vec{s}_m(1)}$.

\begin{lemma}
\label{lemma:scheduling}
For any election $\election = (C, \Vec{s}_m, P, \tie)$, candidates  $W = \set{w_1, ..., w_k} \subseteq C$ and scores $s_1, ..., s_k$ deciding whether $(s_1, ..., s_k) \in \ps_W (P_1)$ can be solved in $\poly(m)$. 
\end{lemma}
\begin{proof}
We use a reduction to a scheduling problem where the tasks are the candidates, every task has execution time of one unit. Each candidate $w_i$ has release time and deadline
\begin{align*}
    & r_{w_i} = \min \set{j \in [n] : \Vec{s}_m(j) = s_i}, \\
    & d_{w_i} = 1+\max \set{j \in [n] : \Vec{s}_m(j) = s_i}
\end{align*}
which ensures that $c_i$ gets the score $s_i$. Other candidates in $C$ have release time 1 and deadline $m+1$. The partial order on the tasks $P_1$ implies that if $(c_j \succ  c_i) \in P_1$ then $c_i$ does not start before $c_j$ is completed. It holds that $(s_1, ..., s_k) \in \ps_W (P_1)$ if and only if the tasks can be scheduled according to all the requirements. A result from \cite{DBLP:journals/siamcomp/GareyJST81} shows an algorithm for scheduling the tasks according to the partial order in $\poly(m)$.
\end{proof}

\begin{lemma}
\label{lemma:dynamic}
For any election $\election = (C, \Vec{s}_m, P, \tie)$ with polynomial scores, fixed $k$ and candidates  $W = \set{w_1, ..., w_k} \subseteq C$ we can find all vectors in $\ps_W (P)$ in polynomial time. 
\end{lemma}
\begin{proof}
First, for every $i \in [n]$, find all vectors in $\text{PS}_W (P_i)$ by going over all vectors $\Vec{v} \in[n \cdot \Vec{s}_m(1)]_0^k$ and checking if $\Vec{v} \in \ps_W (P_i)$ using lemma \ref{lemma:scheduling}. Then, given  $\ps_W (P_1, ..., P_i)$ observe that
\begin{align*}
    & \ps_W (P_1, ..., P_{i+1}) \\
    & = \set{ \Vec{u} + \Vec{v} : \Vec{u} \in \ps_W (P_1, ..., P_i), \Vec{v} \in \ps_W (P_{i+1})}
\end{align*}
where $\Vec{u} + \Vec{v}$ is a point-wise sum of the two vectors $(\Vec{u} + \Vec{v})(j) = \Vec{u}(j)+ \Vec{v}(j)$.
\end{proof}

\begin{theorem}
\label{thm:PolyRuleNecMemk}
For any fixed $k$, $\necmempk$ can be solved in polynomial time under every positional scoring rule with polynomial scores.
\end{theorem}
\begin{proof}
To determine if a candidate $c$ is a necessary $k$-committee member we search for a counterexample, that is, $k$ candidates $w_1, ..., w_k$ and a completion $T$ where $w_1, ..., w_k$ all defeat $c$. For this goal, for every $\set{w_1, ..., w_k} \subseteq C \setminus \set{c}$ we find all elements in $\text{PS}_{\set{c, w_1, ..., w_k}} (P)$ by the algorithm from lemma \ref{lemma:dynamic} and check if it contains some vector of scores for which $w_1, ..., w_k$ all defeat $c$.
\end{proof}

\begin{definition}[Polygamous Matching]
Given a bipartite graph $G = (V \cup U, E)$ and natural numbers $\alpha_u \leq \beta_u$ for all $u \in U$, is there a subset of $E$ where each $v \in V$ is incident to exactly one edge and every $u \in U$ is incident to at least $\alpha_u$ edges and at most $\beta_u$ edges?
\end{definition}
It was shown in \cite{DBLP:conf/pods/KimelfeldKT19} that Polygamous Matching is solvable in polynomial time. For a partial order $P_i$ over $C$, denote by $\max(P_i)$ the candidates which can be in the top position in completions of $P_i$, and the same for the bottom position $\min(P_i)$.

\begin{lemma}
\label{lemma:committeeMatching}
The following decision problem can be solved in polynomial time with respect to plurality and veto: Given $\election = (C, P, \tie)$ and $\alpha'_c, \beta'_c$ for every candidate $c$, is there a completion $T$ such that $\alpha'_c \leq s(T, c) \leq \beta'_c$ for every $c \in C$?
\end{lemma}
\begin{proof}
For both rules the problem is solved by Polygamous Matching. For plurality, $V = \set{v_1, ..., v_n}$ (the set of voters), $U = C$, $E$ connects $u \in U$ and $v_i \in V$ whenever $u \in \max(P_i)$. The bounds are simply $\alpha_c = \alpha'_c, \beta_c = \beta'_c$. For veto, note that receiving a score $s$ is equivalent to being placed at the bottom position at $n-s$ voters. $V,C$ are the same as before, $E$ connects $u \in U$ and $v_i \in V$ whenever $u \in \min(P_i)$.  The bounds are $\alpha_c = n-\beta'_c, \beta_c = n-\alpha'_c$.
\end{proof}

\begin{theorem}
For any fixed $k$, $\posmempk$ can be solved in polynomial time under the rules plurality and veto.
\end{theorem}
\begin{proof}
To determine if a candidate $c$ is a possible $k$-committee member we search for a completions in which $c$ defeats $m-k$ candidates. For every $W \subseteq C \setminus \set{c}$ of size $m-k$ and score $s\leq n$ we use lemma \ref{lemma:committeeMatching} to check if there exists a completion $T$ such that $s(T,c) \geq s$, for every $w \in W$ if $c \tie w$ then $s(T,w) \leq s$, otherwise $s(T,w) < s$.
\end{proof}

\begin{theorem}
\label{thm:PolyRulePosMemk}
For any fixed $k$, $\posmempk$ is NP-complete under every strongly pure positional scoring rule with polynomial scores, other than plurality and veto.
\end{theorem}

\begin{proof}
By reduction from PW with respect to the same scoring rule. Define $m' = m+k-1$. For every $m$, there exists an index $t \leq k-1$ such that 
\begin{align*}
    \Vec{s}_{m'} =& (\Vec{s}_{m'}(1), ..., \Vec{s}_{m'}(t)) \circ \Vec{s}_m \\
    & \circ (\Vec{s}_{m'}(t+m+1), ..., \Vec{s}_{m'}(m'))
\end{align*}
That is, $\Vec{s}_{m'}$ is obtained from $\Vec{s}_m$ by inserting $t$ values at the top coordinates and $k-1-t$ values at the bottom coordinates.
Given $\election^1 = (C^1 = \set{c_1, ..., c_m}, \Vec{s}_m, P^1, \tie^1)$, define $\election^2 = (C^2, \Vec{s}_{m'}, P^2, \Vec{s}_{m'}, \tie^2)$ where
\begin{itemize}
    \item $C^2 = C^1 \cup W_1 \cup W_2$ where $W^1 = \set{w_1, ..., w_t}, W^2 = \set{w_{t+1, ..., w_{k-1}}}$. Denote $W = W^1 \cup W^2$.
        \item $P^2$ is composed of two parts. The first part is $(P^2_1, ..., P^2_n)$, where $P^2_i$ is the same as $P^1_i$ except that candidates in $W^1$ are placed at the top positions and the candidates in $W^2$ are placed at the bottom positions. Formally, $P^2_i = P^1_i \cup P(W^1, C^1, W^2)$. The second part, $Q$, is composed of $n \cdot \Vec{s}_{m'}(1)$ identical copies of the profile
        $$ \set{Q_{i,j}}_{\begin{subarray}{l} i = 1, ..., k-1 \\ j = 1, ..., m\end{subarray}} $$
        For every $i \in [k-1], j \in [m]$, $Q_{i,j} = M_i(W) \circ M_j(C^1)$.
    \item $\tie^2 = O(W) \circ \tie^1$.
\end{itemize}

We show that the candidates of $W$ always defeat all other candidates. Let $T^2$ be a completion of $P^2$. For every $w \in W$, the score of $w$ in each copy in $Q$ is $m \sum_{i=1}^{k-1} \Vec{s}_{m'}(i)$, therefore
\begin{align*}
    s^2(T^2, w) &\geq S(Q, w) = n \cdot \Vec{s}_{m'}(1) \cdot m \sum_{i=1}^{k-1} \Vec{s}_{m'}(i)
\end{align*}
For every $c \in C^1$, the score of $c$ in in each copy in $Q$ is $(k-1) \sum_{i=k}^{m'} \Vec{s}_{m'}(i) \leq m \sum_{i=1}^{k-1} \Vec{s}_{m'}(i)-1$, where the inequality is due to the assumption that $\Vec{s}_{m'}(1) > \Vec{s}_{m'}(m')$. The total score of $c$ in $T^2$ is
\begin{align*}
    s^2(T^2, c) & \leq n \cdot \Vec{s}_{m'}(1) + s(Q, c) \\
    & = n \cdot \Vec{s}_{m'}(1) + n \cdot \Vec{s}_{m'}(1) 
    \cdot (k-1) \sum_{i=k}^{m'} \Vec{s}_{m'}(i) \\
    & \leq n \cdot \Vec{s}_{m'}(1) + n \cdot \Vec{s}_{m'}(1) 
    \cdot \rpar{m \sum_{i=1}^{k-1} \Vec{s}_{m'}(i)-1} \\
    & = n \cdot \Vec{s}_{m'}(1) \cdot m \sum_{i=1}^{k-1} \Vec{s}_{m'}(i) \leq s^2(T^2, w)
\end{align*}
Since $W$ are the first candidates in $\tie^2$, they always defeat the candidates of $C^1$. 

We show for every $c \in C^1$, $c$ is top-$k$-PW in $\election^2$ if and only if $c$ is a top-1-PW in $\election^1$. Let $T^1 = (T^1_1, ... T^1_n)$ be a completion of $P^1$. Observe the completion $T^2 = (T^2_1, ... T^2_n) \circ Q$ of $P^2$ where $T^2_i = O(W_1) \circ T^1_i \circ O(W_2)$. By the claim we get $s^2(T^2, w) \geq s^2(T^2, c)$ for every $c \in C^1, w \in W$ and by the property of $\Vec{s}_{m'}$ we get $s^2(T^2, c) = s^1(T^1, c) + n \cdot \Vec{s}_{m'}(1) \cdot \sum_{i=k}^{m'} \Vec{s}_{m'}(i)$ for every $c \in C^1$. Conversely, given a completion $T^2 = (T^2_1, ... T^2_n) \circ Q$ of $P^2$, define a completion $T^1$ of $P^1$ by removing $W$ from all orders in $(T^2_1, ... T^2_n)$. We get $s^1(T^1, c) = s^2(T^2, c) - n \cdot \Vec{s}_{m'}(1) \cdot \sum_{i=k}^{m'} \Vec{s}_{m'}(i)$ for every $c \in C^1$.
\end{proof}

\section{Necessary and Possible Committees}

\begin{theorem}
\label{thm:NecCom}
$\neccom$ can be solved in $k \cdot \poly(n,m)$ under every positional scoring rule.
\end{theorem}
\begin{proof}
Given $W \subseteq C$ of size $k$ we search for a counterexample, that is, a pair of candidates $w \in W, c \in C \setminus W$ and a completion $T$ for which $\rank(T, c) < \rank(T,w)$. That task of determining whether a candidate $c$ can defeat another candidate $w$ was solved in $\poly(n,m)$ as part of the algorithm from \cite{DBLP:journals/jair/XiaC11} for NW with respect to any positional scoring rule. The only thing we need to change is the last step of the algorithm, which checks if $c$ defeats $c$ is a specific completion. Instead of $S(c) \geq S(w)$ the condition is $(S(c) > S(w)) \lor (S(c) = S(w) \land c \tie w)$.
\end{proof}

\begin{theorem}
\label{thm:PluralityVetoPosCom}
$\poscom$ can be solved in $k \cdot \poly(n,m)$ with under the rules plurality and veto.
\end{theorem}
\begin{proof}
Given $W \subseteq C$ of size $k$, for every $w \in W, s \leq n$ we search for a completion $T$ in which $w$ gets $s$ votes, for every $w' \in W \setminus \set{w}$ it holds that $\rank(T, w') < \rank(T, w)$ and for every $c \in C \setminus W$ it holds that  $\rank(T, w) < \rank(T, c)$. For this goal, we check if there exists a completion $T$ such that:
\begin{itemize}
    \item $s(T, w) = s$.
    \item For every $w' \in W \setminus \set{w}$, if $w' \tie w$ then $s(T, w') \geq s$, otherwise $s(T, w') > s$.
    \item For every $c \in C \setminus W$, if $w \tie c$ then $s(T, c) \leq s$, otherwise $s(T, c) < s$.
\end{itemize}
This is can be done in $\poly(n,m)$ by lemma \ref{lemma:committeeMatching}.
\end{proof}

\begin{theorem}
For any fixed $k$, $\poscompk$ is NP-complete under every strongly pure positional scoring rule with polynomial scores, other than plurality and veto.
\end{theorem}
\begin{proof}
By the same reduction from Theorem \ref{thm:PolyRulePosMemk}. For every $c \in C^1$, $c$ is a PW in $\election^1$ if and only if $W \cup \set{c}$ is a possible $k$-committee in  $\election^2$.
\end{proof}

\section{Other Committee Selection Rules}

So far, we considered decision problems regrading committee selection where the committee is the $k$ highest rated candidates. There are many other committee selection rules such as Condorcet \cite{10.2307/2101455}, Chamberlin-Courant’s rule \cite{10.2307/1957270} and Monroe's rule \cite{monroe_1995}. For these rules, as opposed to the top-$k$ rule, even under complete profiles some problems such as computing the elected committee are hard. However, we show that some problems regarding Condorcet committees can be solved using techniques from our results.

\def\neccc{\mathsf{NecCC}}
\def\poscc{\mathsf{PosCC}}
\def\neccck{\neccc\angs{k}}
\def\poscck{\poscc\angs{k}}

For every voter $i \in [n]$, define a binary relation on $\committees_k \eqdef \set{W \subseteq C, |W|=k}$ by
$$ W \geq_i W' \iff \sum_{w \in W} s(T_i,w) \geq \sum_{w' \in W'} s(T_i,w')$$
$W \in \committees_k$ is a \emph{weak Condorcet $k$-committee} (CC) if
$$ \abs{ \set{i \in [n] : W \geq_i W'}} \geq \abs{ \set{i \in [n] : W' \geq_i W}} $$
for any other $W' \in \committees_k$. A \emph{strong Condorcet $k$-committee} (SCC) is defined similarly, with $>$ instead of $\geq$. We focus on CC, but all the results we discuss also hold for SCC. \cite{DBLP:journals/mss/Darmann13} showed that for varying committee size, deciding whether a set is a CC is coNP-hard under the rules Borda, $t$-approval and $t$-veto (for any fixed $t \geq 2$). However, it can be solved in polynomial time under the rules plurality and veto. Note that for any fixed $k$, CC can be solved in polynomial time under every positional scoring rule, since we can go over all sets in $\committees_k$ and check if some set defeats $W$. As before we can define the problems of necessary and possible Condorcet committees $\neccc$, $\poscc$ and the versions where $k$ is fixed $\neccck$, $\poscck$.

\begin{theorem}
For any fixed $k$, $\neccck$ can be solved in polynomial time under every positional scoring rule with polynomial scores.
\end{theorem}
\begin{proof}
Given $W \in \mathbb{C}_k$ we search for a set $W' \in \mathbb{C}_k \setminus \set{W}$ and a completion $T$ where $W'$ defeats $W$, that is,
$$ \abs{ \set{i \in [n] : W \geq_i W'} } < \abs{ \set{i \in [n] : W' \geq_i W} }$$
For every such set $W'$, initialize $N_W=0, N_{W'} = 0$. For every $i \in [n]$, use lemma \ref{lemma:dynamic} to find all vectors in $\ps_{W \cup W'}(P_i)$ and consider three cases.
\begin{itemize}
    \item If $\ps_{W \cup W'}(P_i)$ contains a vector of scores for which $\sum_{w \in W'} s_w > \sum_{w \in W} s_w$ (there exists a completion of $P_i$ where $W' >_i W$) then increase $N_{W'}$ by 1.
    \item Otherwise, if $\ps_{W \cup W'}(P_i)$ contains a vector of scores for which $\sum_{w \in W'} s_w = \sum_{w \in W} s_w$ (there exists a completion of $P_i$ where $W' =_i W$) then increase both $N_{W'}, N_W$ by 1.
    \item Otherwise, all vectors in $\ps_{W \cup W'}(P_i)$ satisfy $\sum_{w \in W'} s_w < \sum_{w \in W} s_w$ ($W' <_i W$ for all completions of $P_i$). Increase $N_W$ by 1.
\end{itemize}
Finally, there exists a completion where $W'$ defeats $W$ if and only if $N_{W'} > N_W$.
\end{proof}

For plurality and veto, a set $W \in \committees_k$ is a CC if and only if the candidates of $W$ are the $k$ highest rated candidates w.r.t at least one tie-breaking order. Using Theorems \ref{thm:NecCom}, \ref{thm:PluralityVetoPosCom} we get the following result.
\begin{corollary}
$\neccc$ and $\poscc$ can be solved in $k \cdot 
\poly(n,m)$ under the rules plurality and veto.
\end{corollary}

\end{document}

%% file: intro.tex
A central task in social choice is that of \e{winner determination}---how to aggregate the candidate preferences of voters to select the winner. Relevant scenarios may be political elections, document rankings in search engines, hiring dynamics in the job market, decision making in multiagent systems, determination of outcomes in sports tournaments, and so on~\cite{DBLP:reference/choice/2016}. Different voting rules can be adopted for this task. The computational social-choice community has studied in depth the family of the \e{positional scoring rules}, where each voter assigns to each candidate a score based on the candidate’s position in the voter's ranking, and then a winner is a candidate who receives the maximal sum of scores. Famous instantiations of the positional scoring rules include the plurality rule (where a winner is most frequently ranked first), the veto rule (where a winner is least frequently ranked last), their generalizations to $t$-approval and $t$-veto, respectively, and the Borda rule (where the score is the actual position in the reverse order).

The seminal work of Konczak and Lang~\shortcite{Konczak2005VotingPW} has addressed the situation where voter preferences are expressed or known in just a partial manner. The framework is based on the notions of the \e{necessary winners} and \e{possible winners}, who are the candidates that win in every completion, or at least one completion, respectively, of the given partial preferences into complete ones. More precisely, a voting profile consists of a partial order for each voter, and a completion consists of a linear extension for each of the partial orders. Determining the necessary and possible winners is computationally challenging since, conceptually, it involves reasoning about the entire (exponential-size) space of such completions. The complexity of these problems has been thoroughly studied in a series of publications that established a full classification of a general class of positional scoring rules (the ``pure'' scoring rules) into tractable and intractable ones~\cite{DBLP:journals/jcss/BetzlerD10,DBLP:journals/jair/XiaC11,DBLP:journals/ipl/BaumeisterR12}. 

The outcome of an election often goes beyond the single winner to the set of top-$k$ winners.  For example, the top-$k$ winners might be the elected parliament members, the entries of the first page of the search engine, the job candidates to recruit, and the finalists of a sports competition.  In the case of a positional scoring rule, the top-$k$ winners are the candidates who receive the top scores (under some tie-breaking mechanism)~\cite{DBLP:journals/jair/MeirPRZ08}.  Adopting the framework of Konczak and Lang~\shortcite{Konczak2005VotingPW}, in this paper we investigate the computational complexity of determining the necessary and possible top-$k$ winners for incomplete voting profiles and positional scoring rules.

We show that the top-$k$ variant makes the problems fundamentally harder than their top-$1$ counterparts
(necessary and possible winners) when $k$ is given as input.  For example, it is known
that detecting the possible winners is NP-hard for every pure rule, with the exception of plurality and veto where the problem is solvable in polynomial time~\cite{DBLP:journals/jcss/BetzlerD10,DBLP:journals/jair/XiaC11,DBLP:journals/ipl/BaumeisterR12};
we show that in the case of top-$k$, the problem is NP-hard for \e{every} pure rule, \e{including plurality and veto}.  Moreover, tractability of the necessary winners \e{does not} extend to the necessary top-$k$ winners: we show that the detecting whether a candidate is necessarily a top-$k$ winner is coNP-complete for a quite general class of positional scoring rules that include all of the aforementioned ones. We also study the impact of fixing $k$ and establish a  more positive picture: detecting the necessary top-$k$ winners is tractable (assuming that the scores are polynomial in the number of candidates) and detecting the possible the top-$k$ winners is tractable for plurality and veto.

The concept of the top-$k$ winners can be viewed as a special case of \e{multiwinner election} that has been studied mostly in the context of \e{committee selection}. Various utilities have been studied for qualifying selected committee, such as maximizing the number of voters with approved candidates~\cite{DBLP:conf/atal/AzizGGMMW15} and, in that spirit, the Condorcet committees~\cite{DBLP:conf/ijcai/ElkindLS11,DBLP:journals/mss/Darmann13}, aiming at proportional representation via frameworks such as Chamberlin and Courant's~\shortcite{10.2307/1957270} and Monroe's~\shortcite{monroe_1995}, and the satisfaction of fairness and diversity constraints~\cite{DBLP:conf/ijcai/CelisHV18,DBLP:conf/aaai/BredereckFILS18}.

In the case of incomplete voter preferences, the generalization of the problem we study is that of detecting the \e{necessary and possible committee members}. These are interesting and challenging problems in all the variants of committee selection, and we leave them for future investigation.  Note, however, that the problem of determining the elected committee can be intractable even if the preferences are complete~\cite{DBLP:conf/ijcai/ProcacciaRZ07,DBLP:journals/scw/ProcacciaRZ08,DBLP:journals/mss/Darmann13,DBLP:journals/tcs/SkowronYFE15}, in contrast to the top-$k$ winners.
Yet, we show that our results imply the tractability of determining whether a candidate set is a
necessary or possible Condorcet committee in the case of the plurality and veto rules.
The problem of multiwinner determination for incomplete votes has been studied by Lu and Boutilier~\shortcite{DBLP:conf/ijcai/LuB13} in a perspective different from the necessary and possible top-$k$ winners: find a committee that minimizes the maximum objection (or ``regret'') over all possible completions.

%% file: preliminaries.tex
\section{Preliminaries}
We begin with some notation and terminology.

\paragraph{Voting Profiles and Positional Scoring Rules.}
Let $C = \set{c_1, \dots, c_m}$ be the set of \emph{candidates} (or
\emph{alternatives}) and let $V = \set{v_1, \dots, v_n}$ be the set of
\emph{voters}. A \emph{voting profile} $\T = (T_1, \dots, T_n)$ consists
of $n$ linear orders on $C$, where each $T_i$ represents the ranking
of $C$ by $v_i$.

A \e{positional scoring rule} $r$ is a series
$\set{ \Vec{s}_m }_{m \in \natural^+}$ of $m$-dimensional vectors
$\Vec{s}_m = (\Vec{s}_m(1), \dots, \Vec{s}_m(m))$ where
$\Vec{s}_m(1) \geq \dots \geq \Vec{s}_m(m)$ and
$\Vec{s}_m(1) > \Vec{s}_m(m)$. We denote $\Vec{s}_m(j)$ by $r(m,j)$.
Some examples of positional scoring rules include the \emph{plurality}
rule $(1, 0, \dots, 0)$, the \emph{$t$-approval} rule
$(1, \dots, 1, 0, \dots, 0)$ that begins with $t$ ones, the \emph{veto}
rule $(1, \dots, 1, 0)$, the \emph{$t$-veto} rule that ends with $t$
zeros, and the \emph{Borda} rule $(m-1, m-2, \dots, 0)$.

Given a voting profile $\T = (T_1, \dots, T_n)$, the score $s(T_i, c, r)$ that
the voter $v_i$ contributes to the candidate $c$ is $r(m,j)$ where $j$
is the position of $c$ in $T_i$. The score of $c$ in $\T$ is
$s(\T, c, r) = \sum_{i=1}^n s(T_i, c, r)$ or simply $s(\T, c)$ if $r$ is clear from context. The \e{winners} (or \e{co-winners})
are the candidates $c$ with a maximal $s(\T, c)$.

We make standard assumptions about the positional scoring rule
$r$. We assume that $r(m,i)$ is
computable in polynomial time in $m$. We also assume that the numbers
in each $\Vec{s}_m$ are co-prime (i.e., their greatest common divisor
is one).

A positional scoring rule is \emph{pure} if
$\Vec{s}_{m+1}$ is obtained from $\Vec{s}_m$ by inserting a
score at some position,  for all $m > 1$.

\paragraph{Partial Profiles.}

A \e{partial voting profile} $\P = (P_1, \dots, P_n)$ consists of $n$
partial orders on set $C$ of candidates, where each $P_i$ represents
the incomplete preference of the voter $v_i$. A \e{completion} of
$\P = (P_1, \dots, P_n)$ is a complete voting profile
$\T = (T_1, \dots, T_n)$ where each $T_i$ is a completion (i.e., linear
extension) of the partial order $P_i$.

The problems of \e{necessary winners} and \e{possible winners} were
introduced by Konczak and Lang~\shortcite{Konczak2005VotingPW}. Given
a partial voting profile $\P$, a candidate $c \in C$ is a
\emph{necessary winner} if $c$ is a winner in every completion $\T$ of
$\P$, and $c$ is a \emph{possible winner} if there exists a completion
$\T$ of $\P$ where $c$ is a winner. The decision problems associated
to a positional scoring rule $r$ are those of determining, given a
partial profile $\P$ and a candidate $c$, whether $c$ is a necessary
winner and whether $c$ is a possible winner. We denote these problems
by \NW and \PW, respectively.  A classification of the complexity of
these problems has been established in a sequence of publications.
\begin{theorem}[Classification Theorem~\cite{DBLP:journals/jcss/BetzlerD10,DBLP:journals/jair/XiaC11,DBLP:journals/ipl/BaumeisterR12}]
  \NW can be solved in polynomial time for every positional scoring
  rule.
  \PW is solvable in polynomial time for plurality and veto; for all
  other pure scoring rules, it is NP-complete.
\label{thm:classification}
\end{theorem}
In this paper, we aim towards generalizing the Classification Theorem
to determine the necessary and possible \e{top-$k$ winners}, as we
formalize next.

\paragraph{Top-k Winners.}
In principle, a \e{top-$k$ winner} is a candidate that is ranked at
one of the top $k$ places with respect to the sum of scores from the
voters. However, for a precise definition, we need to reason about
\e{ties}. One could adopt several options for being on the top-$k$
winners:
\e{(a)} w.r.t.~\e{at least one} tie-breaking order;
\e{(b)} w.r.t.~\e{every} tie-breaking order;
\e{and (c)} w.r.t.~a tie-breaking order \e{given as input}.
For simplicity of presentation, we adopt the third variation and
assume that the tie-breaking order is given as 
input. Nevertheless, \e{all of our results hold for all three
  variations}. 

Formally, let $r$ be a positional scoring rule, $C$ be a set of
candidates, $\T$ a voting profile, and $\tie$ a \e{tie breaker}, which
is simply a linear order over $C$.  Let $R_{\T}$ be the linear order
on $C$ that sorts the candidates lexicographically by their scores and
then by $\tie$; that is,
\begin{align*}
  R_{\T} \eqdef & \set{c_1 > c_2 : s(\T, c_1) > s(\T, c_2) } \cup \\
                & \set{c_1 > c_2 : s(\T, c_1) = s(\T, c_2)\land c_1 \tie c_2 }\,. 
\end{align*}
A candidate $c$ is a \e{top-$k$ winner} if the position of $c$ in
$R_{\T}$, denoted by $\rank(\T, c)$, is at most $k$. Note that a top-1 winner is necessarily a
winner, but a winner might not be a top-1 winner due to tie breaking.

If $\T$ is replaced with a partial voting profile $\P$, then a candidate $c$ is a \e{necessary top-$k$ winner} if $c$ is a
top-$k$ winner in every completion $\T$ of $\P$, and a
\e{possible top-$k$ winner} if $c$ is a top-$k$ winner in at least one
completion $\T$ of $\P$. Hence, for a positional scoring rule $r$, we
have two computational problems where the input consists
of a candidate set $C$, a partial profile $\P$, a tie breaker $\tie$, a candidate $c$
and a number $k$:
\begin{itemize}
\item In \NTW, the goal is to determine whether $c$ is a necessary
  top-$k$ winner.
\item In \PTW, the goal is to determine whether $c$ is a possible
  top-$k$ winner.
\end{itemize}

We will also consider the versions where $k$ is fixed, and then denote
it by parameterizing the problem with $k$: \NTWk and \PTWk

\paragraph{Additional Notation.}
We use the following notation. For a set $A$ and
a partition $A_1, \dots, A_t$ of $A$:
\begin{itemize}
\item $P(A_1, \dots, A_t)$ denotes the partitioned partial order
$ P(A_1, \dots, A_t) \eqdef \set{a_1 \succ \dots \succ a_t : \forall i
  \in [t], a_i \in A_i}$.
\item $O(A_1, \dots, A_t)$ denotes an arbitrary linear order on $A$ that
  completes $P(A_1, \dots, A_t)$.
\end{itemize}
  A linear order
  $a_1 \succ \dots \succ a_t$ is also denoted as a vector
  $(a_1, \dots, a_t)$. The \e{concatenation} 
  $(a_1, \dots, a_t) \circ (b_1, \dots, b_\ell)$ is $(a_1, \dots, a_t, b_1,\dots, b_\ell)$.

%% file: input-k.tex
\section{Hardness of Top-k Winners}

We first show that the problems we study are computationally
hard for quite general classes of positional scoring rules.

\subsection{Plurality and Veto}
The following theorems
state the hardness for the plurality and veto rules where both \NW and
\PW are solvable in polynomial time (according to the Classification
Theorem).

\begin{theorem}
  \label{thm:PluralityPosNecMem}
  For the plurality rule, \NTW is coNP-complete and \PTW is
  NP-complete.
\end{theorem}

\ifLong
\begin{proof}
  Memberships in the corresponding classes (coNP and NP) are
  straightforward, so we prove only hardness.
  We show a reduction for each of the two
  problems.

  \mypar{\NTW:} We show a reduction from
  \e{exact cover by-3sets} (X3C), which is the following decision
  problem:
  Given a vertex set $U = \set{u_1, \dots, u_{3q}}$ and a collection $E$
  of
  3-element subsets of $U$, can we cover all the
  elements of $U$ using $q$ pairwise-disjoint sets from $E$? For
  $u \in U$, denote by
  $E(u)$ the set $\set{e \in E : u \in e}$ of edges
  incident to $u$.

  Given $U$ and $E$, we construct an
  instance
  $(C, \P, \tie)$ of \NTW under the plurality rule where
  $C = E \cup \set{c^*}$, where $\tie = O(E, \set{c^*})$, and where
  $\P$ is the partial voting profile
  $(P_1, \dots, P_{3q}, T_{3q+1}, T_{3q+2}, T_{3q+3})$.
  For every $i \in [3q]$,
$$ P_i = P (E(u), E \setminus E(u), \set{c^*})\,.$$
This means that the $i$th voter can vote only for edges that cover
$u_i$.
For $ i > 3q$ the order is $T_i = O(\set{c^*}, E)$.
To complete, we show that there is an exact cover if and only if $c^*$
is not a necessary top-$q$ winner.

Suppose that $c^*$ is not a necessary top-$q$ winner, that is, there
are $q$ candidates $e_{i_1}, \dots, e_{i_q}$ and a completion $\T$ of
$\P$ such that $s(\T, e_{i_j}) \geq s(\T,c^*) = 3$ for all
$j \in [q]$. Since every edge can get at most three votes, we get that
$s(\T, e_{i_j}) = 3$ for every $j \in [q]$. Therefore
$e_{i_1}, \dots, e_{i_q}$ is an exact cover: every vertex $u_i$ is
covered by the edge that $P_i$ voted for, and the edges are pairwise
disjoint (since, if two edges are overlapping, then one gets at most
two votes).

Conversely, given an X3C solution $Q = \set{e_{i_1}, \dots, e_{i_q}}$
define a profile $\T = (T_1, \dots, T_{3q+3})$ such that for every
$i \in [3q]$ we have
$$T_i = O(Q \cap E(u_i), E(u_i) \setminus Q, E \setminus E(u_i), \set{c^*})\,.$$
Every $T_i$ extends $P_i$, for all $e \in Q$ we have $s(\T, e) = 3$,
and $s(\T, c^*) = 3$. Hence, all edges in $Q$ defeat $c^*$, and $c^*$
is a not a top-$q$ winner in $\T$.

\mypar{\PTW:} We use a reduction from the \e{dominating set} problem,
which is the following: Given an undirected graph $G = (U,E)$ and an
integer $k$, is there a set $D \subseteq U$ of size $k$ such that
every vertex is either in $D$ or adjacent to some vertex in $D$? Given
a graph $(U,E)$ with $U = \set{u_1, \dots, u_n}$, we construct an
instance $(C, \P, \tie)$ for \PTW under plurality where
$C = U \cup \set{c^*}$, where $\tie = O(\set{c^*}, U)$, and where
$P = (P_1, \dots, P_n)$. Let $N(u_i)$ be the set of neighbours of
$u_i$, and let $N(u_i)^* = N(u_i) \cup \set{u_i}$. For all
$i \in [n]$ we have
$$ P_i \eqdef P(N(u_i)^*, U \setminus N(u_i)^*, \set{c^*})\,. $$
Hence, the $i$th voter can vote only for vertices that dominate $u_i$.
To complete, we show that the graph has a dominating set of size $k$
if and only if $c^*$ is a possible top-$k$ winner.

Suppose there is a dominating set $D$ of size $k$, consider the profile $\T = (T_1, \dots, T_n)$ where for every $i \in [n]$,
$$ T_i \eqdef O(N(u_i)^* \cap D, N(u_i)^* \setminus D, U \setminus N(u_i)^*, \set{c^*})\,.$$
In this completion, for each $u \notin D$ we get $s(\T, u) = 0$. These
are $n-k$ candidates that $c^*$ defeats, therefore $c^*$ is a possible top-$k$ winner. Conversely, if $c^*$ is a possible top-$k$ winner then in some completion $\T$ it defeats at least $n-k$
candidates, and these candidates have a score 0 in $\T$. Let $D$ be the
set of candidates that $c^*$ does not defeat in $\T$, all voters voted
for candidates in $D$ and $|D| \leq k$. A voter $P_i$ can only vote
for vertices which dominate $u_i$, hence $D$ is a dominating set of
size at most $k$.
\end{proof}
\fi

Next, we show the hardness of \NTW and \PTW beyond the plurality rule.
Given a binary positional scoring rule $r$, we define the
\e{complementary-reversed} scoring rule, denoted $r^R$, to be the one
given by $r^R(m,i) = 1-r(m,m+1-i)$. For example, the
complementary-reversed rule of plurality is veto, and more generally,
the complementary-reversed rule of $t$-approval is $t$-veto.

\begin{lemma}
\label{lemma:reversing}
For every binary positional scoring rule $r$, there is a reduction
\begin{enumerate}
\item from \NTW for $r$ to the complement of \PTW for $r^R$;
\item from \PTW  for $r$ to the complement of \NTW  for $r^R$.
\end{enumerate}
\end{lemma}
\ifLong
\begin{proof}
For a partial order $P_i$, the \e{reversed order} is defined by $P_i^R \eqdef \set{x \succ y : (y \succ x) \in P_i}$. Note that $T_i$ extends $P_i$ if and only if $T_i^R$ extends $P_i^R$.

Given $(C, \P, \tie)$ as input under $r$ with $\P = (P_1, ..., P_n)$, consider $(C, \P', \tie^R)$ under $r^R$ where $\P' = (P_1^R, ..., P_n^R)$. Let $\T = (T_1, ..., T_n)$ be a completion of $\P$, observe the completion $\T' = (T_1^R, ..., T_n^R)$ of $\P'$. For every candidate $c$ and a voter $v_i$ we get $s(T_i^R, c, r^R) = 1-s(T_i, c, r)$ so overall $s(\T', c, r^R) = n-s(\T, c, r)$. Since the tie-breaking order is also reversed, it holds that $\rank(\T', c) = m+1-\rank(\T, c)$ for every $c \in C$. In same way, if $\T'$ is a completion of $\P'$ then by reversing the orders we get a completion $\T$ of $\P$ such that $\rank(\T, c) = m+1-\rank(\T', c)$ for every $c \in C$. We can deduce that for any candidate $c$ and integer $k$,
\begin{enumerate}
    \item $c$ is not a necessary top-$k$ winner w.r.t $(C, \P, \tie)$ and $r$ (there exists a completion $\T$ of $\P$ such that $\rank(\T, c) > k$) if and only if $c$ is a possible top-$(m-k)$ winner w.r.t $(C, \P', \tie^R)$ and $r^R$ (there exists a completion $\T'$ of $\P'$ such that $\rank(\T', c) \leq m+1-k$).
    \item $c$ is a possible top-$k$ winner w.r.t $(C, \P, \tie)$ and $r$ (there exists a completion $\T$ of $\P$ such that $\rank(\T, c) \leq k$) if and only if $c$ is not a necessary top-$(m-k)$ winner w.r.t $(C, \P', \tie^R)$ and $r^R$ (there exists a completion $\T'$ of $\P'$ such that $\rank(\T', c) > m-k$).
\end{enumerate}
From the above two points we conclude the two parts of the lemma, respectively.
\end{proof}
\fi
 
Combining Lemma~\ref{lemma:reversing} and
Theorem~\ref{thm:PluralityPosNecMem}, we conclude that:

\begin{theorem}
\label{thm:VetoPosNecMem}
For the veto rule, \NTW is coNP-complete and \PTW is
  NP-complete.
\end{theorem}

\subsection{Beyond Plurality and Veto}
What about positional scoring rules other than plurality and veto? For any other pure positional scoring rule, \PW is NP-complete by the Classification
Theorem, so \PTW is also NP-complete (by choosing $k=1$). Combining
this observation with Theorems~\ref{thm:PluralityPosNecMem} and~\ref{thm:VetoPosNecMem}, we conclude that:
\begin{corollary}
    \label{cor:NecPosMemPure}
    \PTW is NP-complete for every pure positional scoring rule.
\end{corollary}

While we do not have a full classification for \NTW, we show
the hardness of \NTW under general conditions that include
the commonly studied rules. First, we can deduce hardness for every pure positional scoring rule with binary scores. We already established hardness for plurality and veto in Theorems~\ref{thm:PluralityPosNecMem} and \ref{thm:VetoPosNecMem}. For any other rule $r$ in this class, \PW is NP-complete for $r^R$ (by the Classification
Theorem), so a small change in the proof of Lemma~\ref{lemma:reversing} shows that \NTW is coNP-complete for $r$. We conclude that:
\begin{corollary}
    \NTW is coNP-complete for every pure positional scoring rules with binary scores.
\end{corollary}

To discuss rules with scores beyond binary, we define the class of \e{polynomially frequent scoring rules} where some score occurs frequently in the scoring vector. All commonly studied rules fall under this definition, except for Borda.

\begin{definition}
  A positional scoring rule $r$ is \e{polynomially frequent}
  if there exists a score $x$ and a constant $\varepsilon>0$ such that
  $|\set{i : r(m,i) = x}| = \Omega(m^\varepsilon)$.
\end{definition}
Examples of polynomially frequent rules include $t$-approval (where
$x=0$) and $t$-veto (where $x=1$). Another rule is $(2,1,\dots,1,0)$
that has been studied in depth~\cite{DBLP:conf/atal/BaumeisterRR11}.
This class strictly generalizes that of the \e{almost
  constant} scoring rules that has also been studied in the context of
the complexity of winner
determination~\cite{DBLP:conf/ijcai/KimelfeldKS18,DBLP:conf/aaai/KenigK19}.

We will prove that \NTW is hard for all pure polynomially
frequent scoring rules. For that, we need a definition and a lemma.
Let $r$ and $r'$ be two positional scoring rules.
We say that $r'$ \e{polynomially contains} $r$ if there exist a
polynomial $p(m)$, an index $i_m \leq p(m)-m$
and two numbers $a_m>0$ and $b_m$ such that
$r(m,j) = a_m \cdot r'(p(m),i_m+j) + b_m$ for all $m, j \in [m]$. For instance, $t$-approval polynomially contains plurality for every fixed $t$, by choosing $p(m) = m+t-1$, $i_m = t-1$, $a_m = 1$
and $b_m = 0$. Similarly, $t$-veto polynomially contains veto.

\begin{lemma}
  \label{lemma:polyContained}
  Let $r$ and $r'$ be two positional scoring rules.
  If $r'$ polynomially contains $r$, then there is a reduction
  \begin{enumerate}
    \item from \NTW for $r$ to \NTW for $r'$;
    \item from \PTW for $r$ to \PTW for $r'$.
  \end{enumerate}
\end{lemma}
  
\ifLong
  \begin{proof} Let $p(m)$, $i_m$, $a_m$, and $b_m$ be the functions
  that realize the polynomial containment. Given 
  $(C,\P, \tie)$ as input under $r$ with $\P = (P_1, \dots, P_n)$,
  consider
  the input $(C',\P', \tieprime)$ under  $r'$ where:
  \begin{itemize}
  \item $C' = C \cup D_1 \cup D_2$ where $D_1$ and $D_2$ are disjoint
    sets
    of new candidates with
    $ |D_1| = i_m$ and $|D_2| = p(m)-i_m-m$.
    Note that $|C'| = p(m)$.
  \item $\P' = (P'_1, \dots, P'_n)$ where, in each
    $P'_j$, the $i_m$ highest-ranked candidates are the ones of $D_1$, the
    $p(m)-i_m-m$ lowest-ranked candidates are the ones of $D_2$, and
    between
    $D_1$ and
    $D_2$
    the candidates of $C$ are the same as in $P_j$. In our notation,
    $P'_j \eqdef P_j \cup P(D_1, C, D_2)$.
  \item In $\tieprime$, the $i_m$ highest-ranked candidates
    are the ones of $D_1$, the $p(m)-i_m-m$ lowest-ranked candidates
    are the ones of $D_2$, and
    between $D_1$ and $D_2$ the candidates of $C$ are the same as in
    $\tie$. In our notation,
    $\tieprime = O(D_1) \circ \tie \circ O(D_2)$.
  \end{itemize}
  Let $\T =(T_1, \dots, T_n)$ be a completion of
  $\P$. Observe the completion $\T' =(T'_1, \dots, T'_n)$ of $\P'$
  where $T'_j = O(D_1) \circ T_j \circ O(D_2)$. Since $r'$
  polynomially contains in $r$, for every $c \in C$ we get that
  \begin{align*}
      s(\T', c, r') &= \sum_{j=1}^n s(T'_j, c, r') = \sum_{j=1}^n a_m \cdot s(T_j, c, r) + b_m \\
      &= a_m \cdot s(\T, c, r) + n b_m\,.
  \end{align*}
  Hence, by the definition of $\tieprime$, the positions of the candidates satisfy
  $\rank(\T', c) = \rank(\T, c)+i_m$. Conversely, let $\T'$ be a
  completion of $\P'$. Every $T'_j$ has to be of the form
  $O(D_1) \circ O(C) \circ O(D_2)$, so removing $D_1$ and $D_2$ from
  all the linear
  orders gives a completion $\T$ of $\P$ such that
  $\rank(\T, c) = \rank(\T', c) - i_m$ for all $c \in C$. We conclude
  that for all $c \in C$ and $k$ it holds that
  $c$ is a necessary (resp., possible) top-$k$ winner w.r.t.~$(C,\P, \tie)$ and
  $r$ if and only if $c$ is a
  necessary (resp., possible) top-$(k+i_m)$ winner w.r.t.~$(C',\P', \tieprime)$ and $r'$.
  \end{proof}
\fi

Note that Lemma~\ref{lemma:polyContained} can be applied for non-pure rules, so the second item can be used for rules not covered by Corollary~\ref{cor:NecPosMemPure}. Consequently, we get the following general hardness.

\begin{theorem}
    \NTW is coNP-complete for every pure polynomially frequent scoring rule.
\end{theorem}

\ifLong
\begin{proof}
  Let $r$ be a pure polynomially frequent scoring rule, and denote $r$
  by $\set{\Vec{s}_m}_{m>1}$.
  Let $m'$ be the minimal $m$ where $x$ occurs in
  $\Vec{s}_m$, and let $j$ be the index of $x$ in $\Vec{s}_{m'}$.
  First, consider the case where $j > 1$, which means that $\Vec{s}_{m'}$ contains both $x$ and a score greater than $x$. Since the rule is pure we can deduce that for every $m \geq m'$, $\Vec{s}_m$ also contains both $x$ and a score greater than $x$. Then, there exists some score
  $y_m > x$ such that $\Vec{s}_{O(m^{1/\varepsilon})}$ contains the
  vector $(y_m, x, \dots, x)$ of length $m+1$. Choosing
  $p(m) = O(m^{1/\varepsilon})$, $i_m$ as the biggest index of $y_m$, 
  $a_m = y_m-x$ and $b_m = x$ shows that, in this case, $r$ polynomially contains plurality. The hardness results then
  follow from Theorem~\ref{thm:PluralityPosNecMem} and
  Lemma~\ref{lemma:polyContained}.
  Now consider the case where $j = 1$, hence for every $m \geq m'$, $\Vec{s}_m$ contains both $x$ and a score smaller than $x$ (since the rule is pure). Then, there exist some score $y_m < x$ such
  that $\Vec{s}_{O(m^{1/\varepsilon})}$ contains the vector
  $(x, \dots, x, y_m)$ of length $m+1$. Choosing
  $p(m) = O(m^{1/\varepsilon})$, $i_m$ as the smallest index of $x$,
  $a_m = x-y_m, b_m = y_m$ shows that in this case $r$ polynomially contains veto. Hardness for $r$ then follows from
  Theorem~\ref{thm:VetoPosNecMem} and Lemma~\ref{lemma:polyContained}.
\end{proof}
\fi

To complete the picture, we are still missing the complexity of the
top-$k$ winners for the (pure) positional scoring rules that are not
polynomially frequent. An example that stands out is the Borda rule.
This is left as an open direction for future investigation that we
have found quite challenging. For the special case of Borda,
we can prove the hardness of \NTW.

\begin{theorem}
\label{thm:BordaNec}
\NTW is coNP-complete for the Borda rule.
\end{theorem}
The proof, discussed next, is
nontrivial and heavily relies on the specific structure of this rule.

\subsubsection{Proof of Theorem~\ref{thm:BordaNec}}
\def\Scover{S_{\mathsf{cvr}}}

We use the technique of \e{circular voting blocks} of Baumeister, Roos
and J{\"{o}}rg~\shortcite{DBLP:conf/atal/BaumeisterRR11}. For a set
$A = \set{a_1, \dots, a_t}$ and $i \in [t]$, the \e{$i$th circular
  vote} is
\begin{align*}
    M_i(A) \eqdef (a_i, a_{i+1}, \dots, a_t, a_1, a_2, \dots, a_{i-1})\,.
\end{align*}

We reduce from X3C as defined in the proof of Theorem~\ref{thm:PluralityPosNecMem}. Given an X3C instance $(U,E)$ with
$U = \set{u_1, \dots, u_{3q}}$ and $E = \set{e_1, \dots, e_m}$, we
construct an input $(C, \P, \tie)$ for \NTW under Borda.
The candidates are $C = E \cup D \cup \set{c^*}$ where
$D = \set{d_1, \dots,d_{m-1}}$ and $\tie = O(E, \set{c^*}, D)$. The
voting profile is the concatenation (union)
$\P = \P^1 \circ \T^2 \circ \T^3$ of the three parts described
next.

First, $\P^1 = (P_1^1, \dots, P_{3q}^1)$. For every $i \in [3q]$, only
edges that cover $u_i$ can receive a score of $2m-1$. Edges that do
not cover $u_i$ can receive at most $m-1$, and $c^*$ receives $0$.
Formally, denote by $\deg(u_i) = |E(u_i)|$ the degree of $u_i$ in the
graph and let
$D_{\leq j} = \set{d_1, \dots, d_j}$ and $D_{> j} = \set{d_{j+1}, \dots,
  d_{m-1}}$. We assume that $1 \leq \deg(u_i) \leq m-1$, otherwise the
problem is trivial. The partial order is
        $P_i^1 = P(E(u_i), D_{\leq m-\deg(u_i)}, 
        D_{> m-\deg(u_i)} \cup (E \setminus E(u_i)), \set{c^*})$.
        Note that for every $e \in E \setminus E(u_i)$, the number of
    candidates ranked above $e$ in $P_i$ is
    $|E(u_i)| + m-\deg(u_i) = m$, so indeed it receives a score of at
    most $m-1$.

    Second, $\T^2$ is composed of
    $\Scover \eqdef 3(2m-1) + 3 \sum_{i=1}^{q-1} (m-i)$ copies of
    the profile $(T^2_1, \dots, T^2_m)$. For every $i \in [m]$ the
    order $T^2_i$ is constructed in the following way. Start with
    $M_i(E)$, then insert $D$ and $c^*$ such that the score of $c^*$
    is $m$ and scores of $e_i, \dots, e_m, e_1, \dots, e_{i-1}$ are
    $2m-1, 2m-3, \dots, 5, 3, 0$ accordingly if $m$ is even. If $m$ is
    odd then the scores of the edges are the same as before but
    with $m-1$ instead of $m$ and 1 instead of 0. Note that in both
    cases, the sum of the scores of the edges is the same.

    Finally,
    $\T^3$ consists of $4(3q + \Scover)$ copies of the profile
    $(T^3_1, \dots, T^3_{m+1})$. For this part only, denote
    $e_{m+1} = c^*$. For every $i \in [m+1]$,
    $T^3_i = M_i(\set{e_1, \dots, e_{m+1}}) \circ O(D)$.

\ifLong
We state some observations regarding the profile. In $\T^2$, the score of $c^*$ is
\begin{align*}
    s(\T^2, c^*) = \Scover \cdot s((T^2_1, \dots, T^2_m), c^*) =  \Scover \cdot m^2
\end{align*}
For every $e \in E$, the score in $\T^2$ is
\begin{align*}
    s(\T^2, e) &= \Scover \sum_{i=2}^{m} (2i-1) =  \Scover \rpar{m^2 - 1}\\
    &= s(T^2, c^*) - \Scover\,.
\end{align*}
In $\T^3$, for every $d \in D$, the score is
\begin{align*}
    s(\T^3, d) &\leq 4(3q + \Scover)(m+1)(m-2) \leq 4m^2 (3q +\Scover)
\end{align*}
For every $c \in C \setminus D$, the score in $\T^3$ is
\begin{align*}
    & s(\T^3, c) = 4(3q + \Scover) \sum_{i=1}^{m+1} (2m-i) \\
    &= 2(3q + \Scover)(3m^2+m-2) \geq 6m^2 (3q + \Scover)\,.
\end{align*}
Hence, for every pair $d \in D, c \in C \setminus D$ it holds that
$s(\T^3, c) - s(\T^3, d) \geq 2m^2 (3q + \Scover)$. Let
$\T = \T^1 \circ \T^2 \circ \T^3$ be a completion of $\P$. For every
pair $d \in D, c \in C \setminus D$, the score in $\T^1, \T^2$ satisfy $s(\T^1, d) - s(\T^1, c) < 6qm$ and
$s(\T^2, d) - s(\T^2, c) < \Scover \cdot 2m^2$ (by the definition of the Borda rule). Combining these two inequalities with what we showed for $T^3$, we can deduce that $s(\T, c) - s(\T,d) > 0$, which means that the candidates in
$C \setminus D$ always defeat all candidates in $D$.
\fi

\begin{claim}
If there is an X3C, then $c^*$ is not a necessary top-$q$ winner.
\end{claim}
\begin{proof}
  
  Assume, w.l.o.g., that the exact cover is
  $Q = \set{e_1, \dots, e_q}$ and every edge in the cover is
  $e_i = \set{i, q+i, 2q+i}$. Define a completion
  $\T = \T^1 \circ \T^2 \circ \T^3$, $\T^1 = (T^1_1, \dots, T^1_n)$. In
  $\T^1$, every edge in the cover $e_i \in Q$ receives the following scores:

  \begin{itemize}
    \item $e_i$ gets $2m-1$ from $T^1_i, T^1_{q+i}, T^1_{2q+i}$, that is, $e_i$ is placed at the top in the vertices which it covers.
    \item $e_i$ gets $m-1$ from the vertices of $e_{i+1}$, gets $m-2$ from the vertices of $e_{i+2}$ and so on (when we reach $e_q$ we go to $e_1$ and continue until $e_{i-1}$). Note that in this way there is no vertex that should give the same score to two different edges, and since $q \leq m$ the range of scores $e_i$ gets is $\set{m-1, m-2,  \dots, m-q+1} \subseteq \set{m-1, \dots, 1}$ as required by the definition of $\P^1$.
    \end{itemize}  
    For every edge $e \in Q$, the total score in $\T^1$ is
    $s(\T^1, e) = 3(2m-1) + 3 \sum_{i=1}^{q-1} (m-i) = \Scover$ and
    recall that $s(\T^1, c^*) = 0$. Combining this with what we already
    know for $\T^2$ and $\T^3$ implies that \begin{align*}
    s(\T, e) - s(\T, c^*) & = \Scover- \Scover+ 0 = 0\,.
\end{align*}
All candidates in $Q$ defeat $c^*$, hence $c^*$ is not a
necessary top-$q$ winner.
\end{proof}

\begin{claim}
If $c^*$ is not a necessary top-$q$-winner, then there is an X3C.
\end{claim}

\begin{proof} Let $\T = \T^1 \circ \T^2 \circ \T^3$ be a completion
  where at least $q$ candidates defeat $c^*$, let $Q$ be the $q$
  highest rated candidates in $T$. The candidate $c^*$ always defeats
  all candidates
  in $D$, hence $Q \subseteq E$. We show a lower bound and an upper
  bound on the total score of $Q$ in $\T^1$.

  \mypar{Lower bound.} As
  we already showed, every $e \in Q$ should get
  $s(\T^1, e) \geq \Scover$ in order to defeat $c^*$, therefore
  \begin{align*}
    \sum_{e \in Q} s(\T^1, e) &\geq q \cdot \Scover= 3q \cdot \rpar{2m-1 + \sum_{i=1}^{q-1} (m-i)} \\
    & = 3mq^2 + 3mq - \frac{3q^3}{2} + \frac{3q^2}{2} - 3 q
  \end{align*}
  
  \mypar{Upper bound.} For every $u \in U$ denote by $\deg_Q(u)$ the
  degree of $u$ in the sub-graph induced by $Q$, then
  $\sum_{u \in U} \deg_Q(u) = 3q$. We get that
  \begin{align*}
    \sum_{e \in Q} s(\T^1, e) &\leq \sum_{i=1}^{3q} \rpar{\sum_{j=1}^{\deg_Q(u_i)} (2m-j) + \!\!\!\!\sum_{j=1}^{q-\deg_Q(u_i)}\!\!\!\! (m-j)} \\
    = & 3mq^2 + 3mq - \frac{3q^3}{2} + \frac{3q^2}{2} - \sum_{i=1}^{3q} \deg_Q(u_i)^2 \,.
\end{align*}
Overall, the two bounds imply that
$\sum_{i=1}^{3q} \deg_Q(u_i)^2 \leq 3q$ and this is possible only if
all the degrees are one: When all degrees are one, the sum is exactly
$3q$. If we decrease $\deg_Q(u_i)$ to zero and increase $\deg_Q(u_j)$
to two, then the total sum increases by three. Any further changes
cannot decrease the total sum. Therefore, all degrees in the sub-graphs
induced by $Q$ are one, and $Q$ is an X3C.
\end{proof}

\subsection{Top-k Sets}

We have shown hardness results for the necessary and possible top-$k$
winners. Interestingly, we can retain the tractable cases of the
necessary-winner and possible-winner problems for the variant of the
problem where we are given a set $C'\subseteq C$ of $k$ candidates,
and the goal is to determine whether $C'$ constitutes the exact set of
top-$k$ winners. We say that $C'$ is a \e{necessary top-$k$ set} if
$C'$ is the set of top-$k$ winners in every completion, and a
\e{possible top-$k$ set} if $C'$ is the set of top-$k$ winners in at
least one completion.

\begin{theorem}
\label{thm:topkset}
Let $r$ be a positional scoring rule. We can determine in
time $k \cdot \poly(n,m)$:
\begin{enumerate}
  \item whether a given candidate set is a necessary top-$k$
    set;
    \item whether a given candidate set is a possible
      top-$k$ set, assuming that $r$ is either plurality or veto.
\end{enumerate}
\end{theorem}

\ifLong
\begin{proof}
For the first part of Theorem~\ref{thm:topkset} (necessity), we only
need to determine whether a candidate outside of $C'$ can defeat a
candidate from $C'$. This can be done using the algorithm of Xia and
Conitzer~\shortcite{DBLP:journals/jair/XiaC11}, with a minor adjustment to
account for tie breaking.

For the second part (possibility), for every candidate $c' \in C'$ and every integer score $0\leq s\leq n$ we use Lemma~\ref{lemma:committeeMatching} (that we prove in the following
section) to check if there exists a completion $\T$ which satisfies the following conditions. First, $s(\T, c') = s$. Second, for every $c \in C' \setminus \set{c'}$, if $c \tie c'$ then $s(\T, c) \geq s$, otherwise $s(\T, c) > s$. This means that all candidates in $C' \setminus \set{c'}$ defeat $c'$. Finally, For every $c \in C \setminus C'$, if $c' \tie c$ then $s(\T, c) \leq s$, otherwise $s(\T, c) < s$. This means that $c'$ defeats all candidates in $C \setminus C'$. $C'$ is a possible top-$k$ set if and only if such completion exists for some candidate $c' \in C'$ and score $s$.
\end{proof}
\fi

Interestingly, for the plurality and veto rules, a set
$C'$ of candidates is a top-$k$ set for at least one
tie-breaking order if and only if $C'$ is a \e{Condorcet
  committee}~\cite{10.2307/2101455}. Then, by a simple adjustment of
the proof of Theorem~\ref{thm:topkset} we can conclude that, in the
case of plurality and veto, one can determine in polynomial time
whether a given candidate set is a necessary or possible Condorcet
committee.

%% file: bounded.tex
\section{The Case of a Fixed k}
In the previous section, we established that the problems of finding
the necessary and possible top-$k$ winners are very often intractable.
In this section, we investigate the complexity of these problems under
the assumption that $k$ is fixed (and, in particular, can be the degree
of the polynomial that bounds the running time). We will show that the
complexity picture for \NTWk and \PTWk is way more positive, as we
generalize the tractability of \e{almost} all of the tractable scoring
rules for \NW and \PW. We will also generalize hardness results from
\PW to \PTWk; interestingly, this generalization turns out to be quite
nontrivial. 

\subsection{Tractabiliy of \NTWk}
We first prove that \NTWk is tractable for every positional scoring
rule (pure or not), as long as the scores are bounded by a polynomial
in the number $m$ of candidates; in this case, we say that the rule
has \e{polynomial scores}. Note that this assumption is in addition to
our usual assumption that the scores can be computed in polynomial
time. 

\begin{theorem} \label{thm:PolyRuleNecMemk} For all fixed $k$ and
  positional scoring rules $r$ with polynomial scores, \NTWk is 
  in polynomial time.
\end{theorem}

Note that all of the specific rules mentioned so far (i.e.,
$t$-approval, $t$-veto, Borda and so on) have polynomial scores, and
hence, are covered by Theorem~\ref{thm:PolyRuleNecMemk}. An example of
a rule that is \e{not} covered is the rule defined by $r(m,j)=2^{m-j}$.

In the remainder of this section, we prove
Theorem~\ref{thm:PolyRuleNecMemk}. To determine whether a candidate
$c$ is a necessary top-$k$ winner, we search for a counterexample,
that is, $k$ candidates that defeat $c$ in some completion. For that,
we iterate over every subset
$\set{c_1, \dots, c_k} \subseteq C \setminus \set{c}$ and determine
whether these $k+1$ candidates can get a combination of scores that
constitues the counterexample.

More formally, let $C$ be a set of candidates and $r$ a positional
scoring rule. For a partial profile $\P = (P_1, \dots, P_n)$ and a
sequence $S=(c_1,\dots,c_q)$ of candidates from $C$, we denote by
$\ps(\P,S)$ the set of all possible scores that the candidates in $S$
can obtain jointly in a completion:
$$ \ps(\P,S) \eqdef \set{(s(\T,c_1), \dots, s(\T,c_q)) : \T \text{ completes } \P} $$
Note that $\ps(\P,S) \subseteq\set{0,\dots,n \cdot \Vec{s}_m(1)}^q$.
When $\P$ consists of a single voter $P$, we
write $\ps(P,S)$ instead of $\ps(\P,S)$.

A counterexample for $c$ being a necessary top-$k$ winner is a
sequence $S=(c_1,\dots,c_q)$ where $q=k+1$ and $c_q=c$, and a sequence
$(s_1,\dots,s_q)\in\ps(\P,S)$ such that each $c_i$ beats $c$ when for
$i=1,\dots,k$ the score of $c_i$ is $s(c_i)=s_i$ and $s(c)=s_q$. The following two
lemmas show that, indeed, we can find such a counterexample in
polynomial time.

\begin{lemma}
  \label{lemma:scheduling}
  Let $q$ be a fixed number and $r$ a positional scoring rule.
  Whether $(s_1, \dots, s_q) \in \ps(P,S)$ can be determined in
  polynomial time, given a partial order $P$ over a set of candidates,
  a sequence $S$ of $q$ candidates, and scores
  $s_1,\dots, s_q$.
\end{lemma}
\begin{proof}
  We use a reduction to a scheduling problem where tasks have
  \e{execution times}, \e{release times}, \e{deadlines}, and
  \e{precedence constraints} (i.e., task $x$ should be completed
  before starting task $y$).  This scheduling problem can be solved in
  polynomial time~\cite{DBLP:journals/siamcomp/GareyJST81}.  In the
  reduction, each candidate $c$ is a task with a unit execution time.
  For every $c_i$ in $S$, the release time is $\min \set{j \in [n] : r(m,j) = s_i}$,
  and the deadline is
  $1+\max \set{j \in [n] : r(m,j) = s_i}$. For the rest of the candidates, the release time is 1 and the deadline is $m+1$. The precedence constraints are
  $P$. It holds that $(s_1, \dots, s_q) \in \ps(P,S)$ if and only if the
  tasks can be scheduled according to all the requirements.
\end{proof}

From Lemma~\ref{lemma:scheduling} we can conclude that when $q$ is
fixed and $r$ has polynomial scores, we can construct $\ps(\P,S)$ in
polynomial time, via straightforward dynamic
programming.
\begin{lemma}
  \label{lemma:dynamic}
  Let $q$ be a fixed natural number and $r$ a positional scoring rule
  with polynomial scores. The set $\ps(\P,S)$ can be constructed in
  polynomial time, given a partial profile $\P$ and a sequence $S$ of
  $q$ candidates.
\end{lemma}

\ifLong
\begin{proof}
First, for every $i \in [n]$, construct $\ps (P_i, S)$ using Lemma~\ref{lemma:scheduling}. Then, given  $\ps ((P_1, \dots, P_i), S)$, observe that
\begin{align*}
    & \ps ((P_1, \dots, P_{i+1}), S) \\
    & = \set{ \Vec{u} + \Vec{w} : \Vec{u} \in \ps ((P_1, \dots, P_i), S), \Vec{w} \in \ps (P_{i+1}, S)}
\end{align*}
where $\Vec{u} + \Vec{w}$ is a point-wise sum of the two vectors $(\Vec{u} + \Vec{w})(j) = \Vec{u}(j)+ \Vec{w}(j)$. Hence, $\ps(\P,S)$ can be constructed via straightforward dynamic programming.
\end{proof}
\fi

\subsection{Complexity of  \PTWk}

\subsubsection{Plurality and Veto}
We first show that the positional scoring rules that are tractable for
\PW, namely plurality and veto, are also tractable for \PTWk.  This is
done by a reduction to the problem of \e{polygamous
  matching}~\cite{DBLP:conf/pods/KimelfeldKT19}: Given a bipartite
graph $G = (U \cup W, E)$ and natural numbers $\alpha_w \leq \beta_w$
for all $w \in W$, determine whether there is a subset of $E$ where
each $u \in U$ is incident to exactly one edge and every $w \in W$ is
incident to at least $\alpha_w$ edges and at most $\beta_w$ edges.
This problem is known to be solvable in polynomial
time~\cite{DBLP:journals/ipl/Shiloach81,DBLP:conf/aussois/EdmondsJ01}.

\begin{lemma}
\label{lemma:committeeMatching}
The following decision problem can be solved in polynomial time for
the plurality and veto rules: given a partial profile $\P$ over a set
$C$ of candidates and numbers $\gamma_c \leq \delta_c$ for every
candidate $c$, is there a completion $\T$ such that
$\gamma_c \leq s(\T, c) \leq \delta_c$ for every $c \in C$?
\end{lemma}
\begin{proof}
  For both rules, we apply a reduction to polygamous matching, where
  $U=V$ (the set of voters) and $W = C$. For plurality, $E$ connects
  $v_i \in V$ and $c \in C$ whenever $c$ can be in the top position in
  one or more completions of $P_i$, and the bounds are
  $\alpha_c=\gamma_c$ and $\beta_c=\delta_c$. For veto, receiving a
  score $s$ is equivalent to being placed in the bottom position of
  $n-s$ voters, so $E$ connects $v_i \in V$ and $c \in C$ whenever $c$
  can be in the bottom position in one or more completions of $P_i$.
  The bounds are $\alpha_c = n-\delta_c$ and $\beta_c = n-\gamma_c$.
\end{proof}

Finally, to solve \PTWk given $C$, $\P$, $\tie$ and $c$, we consider
every set $D \subseteq C \setminus \set{c}$ of size $m-k$ and search
for a completion where $c$ defeats all candidates of $D$. For that, we
iterate over every integer score $0\leq s\leq n$ and use
Lemma~\ref{lemma:committeeMatching} to test whether there exists a
completion $\T$ such that $s(\T,c) \geq s$, and for every $d \in D$ we
have $s(\T,d) \leq s$ if $c \tie d$ or $s(\T,d) < s$ otherwise. Hence,
we conclude that:
\begin{theorem}\label{thm:ptwk-fixed-possible-veto-plurality}
  For every fixed $k$, \PTWk can be solved in polynomial under the
  plurality and veto rules.
\end{theorem}
The polynomial  degree in
Theorem~\ref{thm:ptwk-fixed-possible-veto-plurality} depends on
$k$. This is unavoidable, at least for the plurality rule, under
conventional assumptions in parameterized complexity.
This is shown by the proof of
Theorem~\ref{thm:PluralityPosNecMem} that gives 
an FPT reduction from the dominating-set problem, which is
$\mathrm{W}[2]$-hard, to $\PTW$.
\begin{theorem}
  \label{thm:plurality-w2-hard}
  Under the plurality rule, 
$\PTW$ is $\mathrm{W}[2]$-hard for the parameter $k$.
\end{theorem}


\subsubsection{Beyond Plurality and Veto}
The Classification Theorem (Theorem~\ref{thm:classification}) states
that \PW is intractable for every pure scoring rule other than
plurality or veto. While this hardness easily generalizes to \PTWk
for $k=1$, it is not at all clear how to generalize it to any
$k>1$. In particular, we cannot see how to reduce \PW to \PTWk while
assuming only the purity of the rule.  We can, however, show such a
reduction under a stronger notion of purity.

A rule $r$ is \e{strongly pure} if the score sequence for $m+1$
candidates is obtained from the score sequence for $m$ candidates by
inserting a new score, \e{either to beginning or the end of the
  sequence}. More formally, $r=\set{ \Vec{s}_m }_{m \in \natural^+}$ is strongly pure
if for all $m \geq 1$, either
$\Vec{s}_{m+1} = \Vec{s}_{m+1}(1) \circ \Vec{s}_m$ or
$\Vec{s}_{m+1} = \Vec{s}_m \circ \Vec{s}_{m+1}(m+1)$.  Note that
$t$-approval, $t$-veto and Borda are all strongly pure.

\begin{figure}[t]
  \small
  \scalebox{0.87}{
  \begin{tabular}{|l|l|c|l|c|l|c|l|}\hline
    scores &  $M_{1,1}$ & $\cdots$ & $M_{1,m'}$ & {\scriptsize $\cdots$}  & $M_{k-1,1}$ & $\cdots$ & $M_{k-1,m'}$\\ \hline
    $r'(1)$ &  $d_1$ &&  $d_1$ && $d_{k-1}$ && $d_{k-1}$\\
    $r'(2)$ & $d_2$ && $d_2$ && $d_1$ && $d_1$\\
    $\vdots$ & $\vdots$ && $\vdots$ &&  $\vdots$ &&  $\vdots$\\ 
    $r'(k-1)$ & $d_{k-1}$ && $d_{k-1}$ && $d_{k-2}$ && $d_{k-2}$\\
    $r'(k)$ & $c_1$ && $c_{m'}$ && $c_1$ && $c_{m'}$\\
    $r'(k+1)$ & $c_2$ && $c_1$ && $c_2$ && $c_1$\\
    $\vdots$ & $\vdots$ && $\vdots$&& $\vdots$ && $\vdots$\\
    $r'(m')$ & $c_{m'}$ && $c_{m'-1}$ && $c_{m'}$ && $c_{m'-1}$\\\hline
  \end{tabular}}
\caption{The voters $M_{i,j}$ used in the proof of
  Theorem~\ref{thm:PolyRulePosMemk} with $r'(j)$ as a shorthand
  notation for $r(m',j)$}
  \end{figure}

\begin{theorem}
  \label{thm:PolyRulePosMemk}
  Suppose that the positional scoring rule is strongly pure, has
  polynomial scores, and is neither plurality nor veto. Then \PTWk
  is NP-complete for all fixed $k$.
\end{theorem}

\ifLong
\begin{proof}
  Let $r$ be a positional scoring rule that satisfies the conditions
  of the theorem, and let us denote $r$ by $\set{\Vec{s}_m}_{m>1}$. We
  use a reduction from \PW under $r$. Consider the input $\P$ and $c$
  for $\PW$ over a set $C$ of $m$ candidates.  Let $m' = m+k-1$. Since
  $r$ is strongly pure, there is an index $t \leq k-1$ such that
\begin{align*}
  \Vec{s}_{m'} =& (\Vec{s}_{m'}(1), \dots, \Vec{s}_{m'}(t)) \circ \Vec{s}_m \\
                & \circ (\Vec{s}_{m'}(t+m+1), \dots, \Vec{s}_{m'}(m'))\,.
\end{align*}
That is, $\Vec{s}_{m'}$ is obtained from $\Vec{s}_m$ by inserting $t$
values at the top coordinates and $k-1-t$ values at the bottom
coordinates.  We define $C'$, $\P'$ and $\tieprime$ as follows.
\begin{itemize}
\item $C' = C \cup D_1 \cup D_2$ where $D_1 = \set{d_1, \dots, d_t}$
  and $D_2 = \set{d_{t+1}, \dots, d_{k-1}}$. Denote
  $D = D_1 \cup D_2$.
\item $\P'$ is the concatenation $\Q\circ\M$ of two partial
  profiles. The first is $\Q=(Q_1, \dots, Q_n)$, where $Q_i$ is the
  same as $P_i$, except that the candidates of $D_1$ are placed at the
  top positions and the candidates of $D_2$ are placed at the bottom
  positions. Formally, $Q_i \eqdef P_i \cup P(D_1, C, D_2)$. The second,
 $\M$, consists of $n \cdot \Vec{s}_{m'}(1)$ copies of the
  profile
  $$ \set{M_{i,j}}_{i = 1, \dots, k-1 \,,\, j = 1, \dots, m} $$
  where $M_{i,j}$ is $M_i(D) \circ M_j(C)$ for the
  circular votes $M_i(D)$ and $M_j(C)$ as defined in the proof of
  Theorem~\ref{thm:BordaNec}.
    \item $\tieprime = O(D, \set{c}, C \setminus \set{c})$.
\end{itemize}

We show that the candidates of $D$ always defeat all other
candidates. For every $d \in D$, the score of $d$ in $\M$ is $s(\M, d) = n \cdot \Vec{s}_{m'}(1) \cdot m \sum_{i=1}^{k-1} \Vec{s}_{m'}(i)$, and for every $c' \in C$ the score in $\M$ is
\begin{align*}
    & s(\M, c') = n \cdot \Vec{s}_{m'}(1) 
    \cdot (k-1) \sum_{i=k}^{m'} \Vec{s}_{m'}(i) \leq \\
    & n \cdot \Vec{s}_{m'}(1) 
    \cdot \rpar{m \sum_{i=1}^{k-1} \Vec{s}_{m'}(i)-1} = s(\M, d) - n \cdot \Vec{s}_{m'}(1)
\end{align*}
where the inequality is due to the assumption that
$\Vec{s}_{m'}(1) > \Vec{s}_{m'}(m')$. Let $\T'$ be a completion of $\P'$, we get that
\begin{align*}
    s(\T', c') &\leq n \cdot \Vec{s}_{m'}(1) + s(\M, c') \\
    & \leq \Vec{s}_{m'}(1) + s(\M, d) - n \cdot \Vec{s}_{m'}(1)  \leq s(\T', d)\,.
\end{align*}
Since the candidates of $D$ are the first candidates in $\tieprime$,
they always defeat the candidates of $C$.

We show that $c$ is a possible winner for $\P$ if and only if $c$ is a possible top-$k$ winner for $(C', \P', \tieprime)$. Let $\T = (T_1, \dots T_n)$ be a completion of $\P$ where $c$ is a winner. Consider the
completion $\T' = (T'_1, \dots T'_n) \circ \M$ of $\P'$ where
$T'_i = O(D_1) \circ T_i \circ O(D_2)$. For every $c' \in C$, we know that
$s(\T', d) \geq s(\T', c')$ for every $d \in D$, and
from the property of $\Vec{s}_{m'}$ we get that
$$ s(\T', c') = s(\T, c') + n \cdot \Vec{s}_{m'}(1) \cdot
\sum_{i=k}^{m'} \Vec{s}_{m'}(i)\,. $$
From the choice of $\tieprime$, $c$ defeats all candidates of $C \setminus \set{c}$ in $\T'$, hence $c'$ is a top-$k$ winner in $\T'$.
Conversely,
let $\T' = (T'_1, \dots T'_n) \circ \M$ be a completion of $\P'$ where $c$ is a top-$k$ winner, define a completion $\T$ of $\P$ by removing $D$ from all orders in $(T'_1, \dots T'_n)$. For every $c' \in C$ we have
$$ s(\T, c') = s(\T', c') - n \cdot \Vec{s}_{m'}(1) \cdot
\sum_{i=k}^{m'} \Vec{s}_{m'}(i) $$
hence $c$ is a winner in $\T$.
\end{proof}
\fi

The proof of Theorem~\ref{thm:PolyRulePosMemk} can be easily adjusted
to show the hardness of determining whether a given candidate set of
the fixed size $k$ is a top-$k$ set.

\begin{theorem}
  \label{thm:PolyRulePosTopkSetFixedK}
  Suppose that the positional scoring rule is strongly pure, has
  polynomial scores, and is neither plurality nor veto. Then for every
  fixed $k$ it is NP-complete to decide whether a given candidate set
  is a possible top-$k$ set.
\end{theorem}

%% file: conclusions.tex
\section{Concluding Remarks}
We studied the problems of detecting the necessary and possible top-$k$ winners over incomplete voting profiles. We showed that these problems are fundamentally harder than their classic top-$1$ counterparts (necessary and possible winners) when $k$ is given as part of the input. For a fixed $k$, we have generally recovered the tractable positional scoring rules of the top-$1$ variant. Many problems are left for investigation in future research: completing our results towards full classifications (of the class of pure rules), establishing useful tractability conditions for an input $k$, further investigating the parameterized complexity of the problem when $k$ is the parameter, detecting the necessary and possible committee members under different committee-selection policies, and incorporating fairness and diversity constraints.